\documentclass[11pt,a4paper]{article}

\usepackage{lipsum}
\usepackage{amsfonts}
\usepackage{graphicx}
\usepackage{epstopdf}
\usepackage{algorithmic}

\ifpdf
  \DeclareGraphicsExtensions{.eps,.pdf,.png,.jpg}
\else
  \DeclareGraphicsExtensions{.eps}
\fi

\usepackage{amsmath}
\usepackage{amssymb}
\usepackage{amsfonts}
\usepackage{amsthm}
\usepackage[utf8]{inputenc}
\usepackage[T1]{fontenc}
\usepackage{xspace}

\usepackage{thm-restate}
\usepackage{hyperref}
\usepackage{enumitem}
\usepackage{subcaption}
\usepackage{mdframed}
\usepackage{tabularx, environ}
\usepackage{cite}
\usepackage{tcolorbox}
\usepackage{cleveref}
\crefformat{footnote}{#2\footnotemark[#1]#3}

\usepackage[affil-sl]{authblk} 

\definecolor{dark-red}{rgb}{0.4,0.15,0.15}
\definecolor{dark-blue}{rgb}{0.15,0.15,0.4}
\definecolor{medium-blue}{rgb}{0,0,0.5}
\definecolor{gray}{rgb}{0.5,0.5,0.5}
\definecolor{color-Ig}{rgb}{0.15,0.7,0.15}

\hypersetup{
colorlinks = true,
urlcolor={medium-blue},
linkcolor = black!30!blue,
citecolor = black!30!green
}

\newtheorem{theorem}{Theorem}[section]

\newtheorem{corollary}[theorem]{Corollary}

\newtheorem{proposition}[theorem]{Proposition}

\newtheorem{lemma}[theorem]{Lemma}

\newtheorem*{definition*}{Definition}

\newtheorem{definition}[theorem]{Definition}

\usepackage{tikz}

\usetikzlibrary{mindmap,positioning}
\usetikzlibrary{graphs}
\usetikzlibrary[graphs]
\usetikzlibrary{patterns}
\usetikzlibrary{arrows,automata,positioning}
\usetikzlibrary{decorations.pathreplacing}
\usetikzlibrary{decorations}
\usetikzlibrary{decorations.pathmorphing}
\usetikzlibrary{shapes.geometric}
\usetikzlibrary{quotes}
\usetikzlibrary{shapes.callouts}
\usetikzlibrary{arrows.meta}
\usetikzlibrary{calc}
\usetikzlibrary{snakes}
\usetikzlibrary{fit}

\tikzset{faded/.style={gray,very thin}}
\tikzset{vertex/.style={draw,circle,minimum size=10pt,inner sep=0pt}}
\tikzset{novertex/.style={circle,minimum size=10pt,inner sep=0pt}}
\tikzset{blackvertex/.style={draw,circle,minimum size=10pt,inner sep=0pt, fill=black}}
\tikzset{redvertex/.style={draw,circle,minimum size=10pt,inner sep=0pt, fill=red}}
\tikzset{redvertexfaded/.style={draw,circle,faded,minimum size=10pt,inner sep=0pt, fill=red!50}}
\tikzset{greenvertex/.style={draw,circle,minimum size=10pt,inner sep=0pt, fill=green}}
\tikzset{greenvertexfaded/.style={draw,circle,faded,minimum size=10pt,inner sep=0pt, fill=green!50}}
\tikzset{bluevertex/.style={draw,circle,minimum size=10pt,inner sep=0pt, fill=blue}}
\tikzset{bluevertexfaded/.style={draw,circle,faded,minimum size=10pt,inner sep=0pt, fill=blue!50}}
\tikzset{yellowvertex/.style={draw,circle,minimum size=10pt,inner sep=0pt, fill=yellow}}
\tikzset{yellowvertexfaded/.style={draw,circle,faded,minimum size=10pt,inner sep=0pt, fill=yellow!50}}

\tikzset{edge/.style = {->,> = latex'}}

\tikzset{snake it/.style={decorate, decoration=snake}}

\newcounter{c}

\makeatletter

\newcolumntype{\expand}{}
\long\@namedef{NC@rewrite@\string\expand}{\expandafter\NC@find}

\NewEnviron{boxproblem}[2][]{%
  \def\boxproblem@arg{#1}%
  \def\boxproblem@framed{framed}%
  \def\boxproblem@lined{lined}%
  \def\boxproblem@doublelined{doublelined}%
  \ifx\boxproblem@arg\@empty%
    \def\boxproblem@hline{}%
  \else%
    \ifx\boxproblem@arg\boxproblem@doublelined%
      \def\boxproblem@hline{\hline\hline}%
    \else%
      \def\boxproblem@hline{\hline}%
    \fi%
  \fi%
  \ifx\boxproblem@arg\boxproblem@framed%
    \def\boxproblem@tablelayout{|>{\bfseries}lX|c}%
    \def\boxproblem@title{\multicolumn{2}{|l|}{%
        \raisebox{-\fboxsep}{\textsc{\normalsize #2}}%
      }}%
  \else
    \def\boxproblem@tablelayout{>{\bfseries}lXc}%
    \def\boxproblem@title{\multicolumn{2}{l}{%
        \raisebox{-\fboxsep}{\textsc{\normalsize #2}}%
      }}%
  \fi%
  \bigskip\par\noindent%
  \begin{tabularx}{\textwidth}{\expand\boxproblem@tablelayout}%
    \boxproblem@hline%
    \boxproblem@title\\[2\fboxsep]%
    \BODY\\\boxproblem@hline%
  \end{tabularx}%
  \medskip\par%
}
\makeatother

\usepackage[normalem]{ulem}
\DeclareMathOperator{\Ocal}{\mathcal{O}}

\newcommand{\RunTimeAlg}{2^{\Ocal(k \log k)} \cdot n^{\Ocal(1)}}

\DeclareMathOperator{\bn}{\text{bn}}
\DeclareMathOperator{\hn}{\text{hn}}
\DeclareMathOperator{\wlink}{\text{wlink}}

\DeclareMathOperator{\order}{\text{{\sf ord}}}

\newcommand{\wdth}{{\sf width}\xspace}
\renewcommand{\bn}{{\sf bn}\xspace}
\renewcommand{\hn}{{\sf hn}\xspace}
\renewcommand{\wlink}{{\sf wlink}\xspace}

\newcommand{\balsep}{balanced separator\xspace}

\newcommand{\dtw}{\text{{\sf dtw}}}
\newcommand{\ord}{\text{{\sf ord}}}
\newmdtheoremenv{box_problem}{Problem}

\renewenvironment{abstract}
 {\small
  \begin{center}
  \bfseries \abstractname\vspace{-.5em}\vspace{0pt}
  \end{center}
  \list{}{%
    \setlength{\leftmargin}{5mm}
    \setlength{\rightmargin}{\leftmargin}%
  }%
  \item\relax}
 {\endlist}

\title{Adapting the Directed Grid Theorem\\into an FPT Algorithm\thanks{Work supported by projects CAPES/STIC-AmSud 88881.197438/2018-01,
CNPq - Universal project 425297/2016-0, FUNCAP - PRONEM PNE-011200061.01.00/16, DEMOGRAPH (ANR-16-CE40-0028), ESIGMA (ANR-17-CE23-0010), ELIT (ANR-20-CE48-0008-01), and UTMA (ANR-20-CE92-0027). An extended abstract of this paper appeared in the \emph{Proc. of the X Latin and American Algorithms, Graphs and Optimization Symposium (\textbf{LAGOS}), volume 346 of ENTCS, pages 229--240, Belo Horizonte, Brazil, June \textbf{2019}.}}}

\author[1]{Victor Campos}
\author[1,2]{Raul Lopes}
\author[1]{Ana Karolinna Maia}
\author[2]{Ignasi Sau}

\affil[1]{\emph{\small ParGO group, Universidade Federal do Cear\'a, Fortaleza, Brazil}}
\affil[2]{\emph{\small LIRMM, Universit\'e de Montpellier, CNRS, Montpellier, France\newline \texttt{\{campos,karolmaia\}@lia.ufc.br, raulwtlopes@gmail.com, ignasi.sau@lirmm.fr}}}

\date{}

\usepackage{amsopn}

\ifpdf
\hypersetup{
  pdftitle={Adapting the Directed Grid Theorem into an FPT algorithm},
  pdfauthor={V. Campos, A. K. Maia, R. Lopes, and I. Sau}
}
\fi

\usepackage{fancyhdr}

\fancypagestyle{page-with-title-authors-header}{%
  \fancyhf{}
  \fancyhead[L]{%
  \ifodd\value{page}
  \centering\textsc{V. Campos, A. K. Maia, R. Lopes, and I. Sau}%
  \else
  \makebox[0pt][l]{\parbox[b]{\headwidth}{\centering \textsc{Adapting the Directed Grid Theorem into an FPT algorithm}}}%
  \fi}
  \fancyfoot[C]{\thepage}
}

\begin{document}

\maketitle

\vspace{-1cm}
\begin{abstract}
The Grid Theorem of Robertson and Seymour~[JCTB, 1986] is one of the most important tools in the field of structural graph theory, finding numerous applications in the design of algorithms for undirected graphs. An analogous version of the Grid Theorem in digraphs was conjectured by Johnson et al.~[JCTB, 2001], and proved by Kawarabayashi and Kreutzer~[STOC, 2015].
Namely, they showed that there is a function $f(k)$ such that every digraph of directed tree-width at least $f(k)$ contains a cylindrical grid of order $k$ as a butterfly minor, and stated that their proof
can be turned into an {\sf XP} algorithm, with parameter $k$, that either constructs a decomposition of the appropriate width, or finds the claimed large cylindrical grid as a butterfly minor.
In this paper, we adapt some of the steps of the proof of Kawarabayashi and Kreutzer to  improve this {\sf XP} algorithm into an {\sf FPT} algorithm. Towards this, our main technical contributions are two \textsf{FPT} algorithms with parameter $k$.
The first one either produces an arboreal decomposition of width $3k-2$ or finds a haven of order $k$ in a digraph $D$, improving on the original result for arboreal decompositions by Johnson et al.~[JCTB, 2001].
The second algorithm finds a well-linked set of order $k$ in a digraph $D$ of large directed tree-width.
As tools to prove these results, we show how to solve a generalized version of the problem of finding balanced separators for a given set of vertices $T$ in \textsf{FPT} time with parameter $|T|$, a result that we consider to be of its own interest.

\end{abstract}



\section{Introduction}
\label{sec:intro}

Width parameters can be seen as an estimation of how close a given graph is to a typical structure. For example, the \emph{tree-width} of a graph, a parameter of particular interest in the literature, measures how tightly a graph can be approximated by a tree. Namely, a \emph{tree decomposition} of a graph $G$ with bounded tree-width shows how one can place the vertices of the original graph into ``bags'' of bounded size which, in turn, can be arranged as the vertices of a tree $T$ such that the intersection between adjacent bags in $T$ are separators in $G$.
Thus, a tree decomposition exposes a form of global connectivity measure for graphs: as only a bounded number of vertices can be placed in each bag, many small separators can be identified through the decomposition.
The tree-width of graphs was first introduced by Bertele and Brioschi~\cite{Umberto.Francesco.72}, then again by Halin~\cite{Halin1976}, and finally reintroduced by Robertson and Seymour~\cite{Robertson.Seymour.86}.
For a survey on the subject, we refer the reader to~\cite{10.1007/11917496_1}.

A number of hard problems can be efficiently solved in graphs of bounded tree-width, either by making use of classical algorithmic techniques like dynamic programming, or by making use of Courcelle's Theorem~\cite{COURCELLE199012}.
Applications of algorithms based on tree decompositions range from frequency allocation problems to the \textsc{Traveling Salesman} problem \cite{Koster.Hoesel.Kolen.99,Cook.Seymour.2003}.

Given the enormous success achieved by applications based on width parameters in undirected graphs, it is no surprise that there is interest in finding similar definitions for digraphs.
Johnson et al.~\cite{Johnson.Robertson.Seymour.Thomas.01} proposed an analogous measure for tree-width in the directed case.
The \emph{directed tree-width} of a digraph measures its distance to being a directed acyclic graph (DAG for short), and an \emph{arboreal decomposition} exposes a (strong) connectivity measure of a digraph.
Reed~\cite{Reed.99} provided an intuitive exposition of the similarities between the undirected and directed cases.

Similarly to the undirected case, some hard problems become tractable when restricted to digraphs of bounded directed tree-width.
For example, Johnson et al.~\cite{Johnson.Robertson.Seymour.Thomas.01} showed that the \textsc{Directed $k$-Disjoint Paths} problem, which Fortune et al.~\cite{FORTUNE1980111} showed to be {\sf NP}-hard even for $k=2$ in general digraphs, is solvable in polynomial (more precisely, in {\sf XP}) time in digraphs of directed tree-width bounded by a constant.
A similar approach given in~\cite{Johnson.Robertson.Seymour.Thomas.01} can be applied to the \textsc{Hamilton Path} and \textsc{Hamilton Cycle} problems, \textsc{Hamilton Path with Prescribed Ends}, and others.
It is worth mentioning that Slivkins~\cite{Slivkins.03} proved that the \textsc{Directed $k$-Disjoint Paths} problem is {\sf W}$[1]$-hard even when restricted to DAGs.
As DAGs have directed tree-width zero, there is little hope for the existence of a fixed-parameter tractable ({\sf FPT} for short) algorithm for the \textsc{Directed $k$-Disjoint Paths} problem in digraphs of bounded directed tree-width.
As another example of application, a Courcelle-like theorem for directed tree-width was proved by de Oliveira Oliveira~\cite{Oliveira16}, but running in \textsf{XP} time.

It is natural to ask what can be said of a graph with large tree-width.
One of the most relevant results in structural graph theory states that undirected graphs with large tree-width contain large grid minors.
More precisely, the Grid Theorem by Robertson and Seymour~\cite{Robertson.Seymour.86} states that there is a function $f: \mathbb{N} \to \mathbb{N}$ such that every graph of tree-width at least $f(k)$ contains a $(k \times k)$-grid as a minor.
Recently, Chekuri and Chuzhoy~\cite{Chekuri:2016:PBG:3015772.2820609} gave a polynomial bound on the function $f(k)$, which was further improved by Chuzhoy and Tan~\cite{ChuzhoyT19}.

Sometimes, large tree-width (and therefore, the existence of a large grid minor) implies that we are actually working with a positive instance of a particular problem.
In this direction, Demaine et al.~\cite{Demaine:2005:SPA:1101821.1101823} presented a framework that generates \textsf{FPT} algorithms for many such problems, known as \emph{bidimensional} problems.
This list includes \textsc{Vertex Cover}, \textsc{Feedback Vertex Set}, \textsc{Longest Path}, \textsc{Minimum Maximal Matching}, \textsc{Dominating Set}, \textsc{Edge Dominating Set}, and many others.
This seminal work is currently known as \emph{Bidimensionality}~\cite{FominDHT16}.

Another application of the Grid Theorem is in the \emph{irrelevant vertex} technique, introduced by Robertson and Seymour~\cite{ROBERTSON199565,ROBERTSON2009583,ROBERTSON2012530} to solve the $k$-\textsc{Disjoint Paths} problem.
The goal is to show that every instance whose input graph violates a set of conditions contains a vertex that is ``irrelevant'', that is, a vertex whose removal generates an equivalent instance of the problem.
This leads to an iterative algorithm, reducing the problem to a smaller instance, until it satisfies sufficient conditions for its tractability.
This technique was used to solve the $k$-\textsc{Disjoint Paths} problem in \textsf{FPT} time with parameter $k$, and a number of other  problems (cf. for instance~\cite{Grohe:2011:FTS:1993636.1993700,Kleinberg:1998:DAU:276698.276867}).
For the directed case, Cygan et al.~\cite{6686155} used a similar technique to provide an {\sf FPT} algorithm for the \textsc{Directed $k$-Disjoint Paths} problem in planar digraphs.

A result analogous to the Grid Theorem for digraphs was conjectured by Johnson et al.~\cite{Johnson.Robertson.Seymour.Thomas.01} and Reed~\cite{Reed.99}, and recently proved by Kawarabayashi and Kreutzer~\cite{Kawarabayashi:2015:DGT:2746539.2746586}\footnote{The full version of~\cite{Kawarabayashi:2015:DGT:2746539.2746586} is available at \href{https://arxiv.org/abs/1411.5681v2}{https://arxiv.org/abs/1411.5681v2}.}, after having proved it for digraphs with forbidden minors~\cite{KawarabayashiK14}\footnote{In an unpublished manuscript from 2001~\cite{johnson2015excluding}, Johnson, Robertson, Seymour and Thomas gave a proof of this result for planar digraphs.}.
Namely, it is shown in~\cite{Kawarabayashi:2015:DGT:2746539.2746586} that there is a function $f: \mathbb{N} \to \mathbb{N}$ such that every digraph of directed tree-width at least $f(k)$ contains a \emph{cylindrical grid} (see Figure~\ref{fig:cylindrical_grid_intro}) of order $k$ as a \emph{butterfly minor}; all the definitions are given formally in Section~\ref{section:formal_def}.
Recently, Hatzel et al.~\cite{HatzelKK19} proved that the function $f(k)$ can be made polynomial in {\sl planar} digraphs.


\begin{figure}[h]
\centering
\scalebox{.5}{
\begin{tikzpicture}

\def \n {8}
\def \margin {3}
\def \radiusinner {1cm}
\foreach \i in {1,...,\n}
{
  \setcounter{c}{\i};
  \def \vname {\alph{c}};

  \node[blackvertex,vertex,scale=.4] (P-\vname) at ({360/\n * (\i - 1)}:\radiusinner) {};
  \draw[-{Latex[length=2mm, width=2mm]}] ({360/\n * (\i - 1)+\margin}:\radiusinner)
  arc ({360/\n * (\i - 1)+\margin}:{360/\n * (\i)-\margin}:\radiusinner);
}

\def \radiusinner {2cm}

\foreach \i in {1,...,\n}
{
  \node[blackvertex,vertex,scale=.4] at ({360/\n * (\i - 1)}:\radiusinner) {};
  \draw[-{Latex[length=2mm, width=2mm]}] ({360/\n * (\i - 1)+\margin}:\radiusinner)
  arc ({360/\n * (\i - 1)+\margin}:{360/\n * (\i)-\margin}:\radiusinner);
}

\def \radiusinner {3cm}

\foreach \i in {1,...,\n}
{

  \node[blackvertex,vertex,scale=.4] at ({360/\n * (\i - 1)}:\radiusinner) {};
  \draw[-{Latex[length=2mm, width=2mm]}] ({360/\n * (\i - 1)+\margin}:\radiusinner)
  arc ({360/\n * (\i - 1)+\margin}:{360/\n * (\i)-\margin}:\radiusinner);
}

\def \radiusinner {4cm}

\foreach \i in {1,...,\n}
{
  \setcounter{c}{\i};
  \def \vname {\Alph{c}};

  \node[blackvertex,vertex,scale=.4] (P-\vname) at ({360/\n * (\i - 1)}:\radiusinner) {};
  \draw[-{Latex[length=2mm, width=2mm]}] ({360/\n * (\i - 1)+\margin}:\radiusinner)
  arc ({360/\n * (\i - 1)+\margin}:{360/\n * (\i)-\margin}:\radiusinner);
}

\draw[-{Latex[length=2mm, width=2mm]}] (P-a) to (P-A);
\draw[-{Latex[length=2mm, width=2mm]}] (P-B) to (P-b);
\draw[-{Latex[length=2mm, width=2mm]}] (P-c) to (P-C);
\draw[-{Latex[length=2mm, width=2mm]}] (P-D) to (P-d);
\draw[-{Latex[length=2mm, width=2mm]}] (P-e) to (P-E);
\draw[-{Latex[length=2mm, width=2mm]}] (P-F) to (P-f);
\draw[-{Latex[length=2mm, width=2mm]}] (P-g) to (P-G);
\draw[-{Latex[length=2mm, width=2mm]}] (P-H) to (P-h);
\end{tikzpicture}
}%
\caption{A cylindrical grid of order $k = 4$.}
\label{fig:cylindrical_grid_intro}
\end{figure}
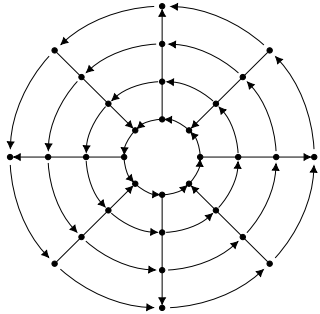
The Directed Grid Theorem has found many applications. For instance, Amiri et al.~\cite{DBLP:journals/corr/AmiriKKW16} proved that a strongly connected digraph $H$ has the Erd\H{o}s-Pósa property if and only if $H$ is a butterfly minor of some cylindrical grid of sufficiently large order.
Additionally, the authors showed that for every fixed strongly connected digraph $H$ satisfying those conditions and every fixed integer $k$, there is a polynomial-time algorithm that either finds $k$ disjoint (butterfly or topological) models of $H$ in a digraph $D$ or a set $X \subseteq V(D)$ of size bounded by a function of $k$ such that $D \setminus X$ does not contain a model of $H$.

Edwards et al.~\cite{edwards_et_al:LIPIcs:2017:7876} applied some results used in the proof of the Directed Grid Theorem~\cite{Kawarabayashi:2015:DGT:2746539.2746586} to provide an \textsf{XP} algorithm with parameter $k$ for a relaxed version of the \textsc{Directed Disjoint Paths} problem, in which every vertex of the input digraph is allowed to occur in at most two paths of a solution, when restricted to $(36k^3 + 2k)$-strongly connected digraphs.
Kawarabayashi and Kreutzer~\cite{Kawarabayashi:2015:DGT:2746539.2746586} mentioned that the Directed Grid Theorem can be used to provide, for fixed $k$, an algorithm running in polynomial time that, given a digraph $D$ and $k$ terminal pairs $(s_1, t_1), \ldots,(s_k, t_k)$, either finds a collection of paths $P_1, \ldots, P_k$ such that $P_i$ is a path from $s_i$ to $t_i$ in $D$ and every vertex of $D$ occurs in at most four paths of the collection, or concludes that $D$ does not contain a collection of pairwise disjoint paths $P_1, \ldots, P_k$ such that $P_i$ is a path from $s_i$ to $t_i$ in $D$, for $i \in [k]$.
Although Chekuri et al.~\cite{chekuri_et_al:LIPIcs:2016:6273} could not use the Directed Grid Theorem since the bound on $f(k)$ (mentioned above) is larger than required, they build on the ideas used in~\cite{johnson2015excluding} to produce their own  version of the Directed Grid Theorem for planar digraphs.

The proof of the Directed Grid Theorem by Kawarabayashi and Kreutzer~\cite{Kawarabayashi:2015:DGT:2746539.2746586} is constructive.
Namely, the authors start with an algorithm by Johnson et al.~\cite[3.3]{Johnson.Robertson.Seymour.Thomas.01}  that, given a digraph $D$ and an integer parameter $k$, outputs, in {\sf XP} time, either an arboreal decomposition of $D$ of width at most $3k-2$ or a haven of order $k$ (see Definition~\ref{def:havens_in_digraphs}).
Thus, if $D$ has directed tree-width at least $3k-1$, they obtain a haven of order $k$.
From this haven, they obtain a bramble $\mathcal{B}$ of order $k$ and size $|V(D)|^{\Ocal(k)}$. 
Finally, from $\mathcal{B}$ they find a path $P$ containing a well-linked set $A$ (see Definition~\ref{def:well-linked-sets-in-digraphs}) of size roughly $\sqrt{k}$ in \textsf{XP} time with parameter $k$.

We remark that the bound on the running time of those algorithms depends on the size of $\mathcal{B}$ since, in general, one must test whether $X \cap V(B) \neq \emptyset$ for each $B \in \mathcal{B}$ to check whether a given set $X \subseteq V(D)$ is a hitting set of $\mathcal{B}$.
The remainder of the proof of the Directed Grid Theorem~\cite{Kawarabayashi:2015:DGT:2746539.2746586} runs in {\sf FPT} time, with parameter~$k$.

\medskip
\noindent \textbf{Our approach, results, and techniques.}
By making local changes to the proofs by Johnson et al.~\cite{Johnson.Robertson.Seymour.Thomas.01} and Kawarabayashi and Kreutzer~\cite{Kawarabayashi:2015:DGT:2746539.2746586}, we show that there is an {\sf FPT} algorithm that, given a digraph $D$ and an integer $k$, either constructs an arboreal decomposition of $D$ of width at most $3k-2$, or finds a path $P$ in $D$ containing a well-linked set $A$ of size roughly $\sqrt{2k}$.
Our results and the remainder of the proof of the Directed Grid Theorem~\cite{Kawarabayashi:2015:DGT:2746539.2746586} yield an \textsf{FPT} algorithm that either constructs an arboreal decomposition of width at most $f(k)$ or a cylindrical grid of order $k$ as a butterfly minor of $D$. For completeness, we provide in Section~\ref{section:finding_a_cylindrical_grid}  an overview of how, starting from the path $P$ and the well-linked set $A$ found by our \textsf{FPT} algorithm, the proof of
Kawarabayashi and Kreutzer~\cite{Kawarabayashi:2015:DGT:2746539.2746586} yields an algorithm to find the desired  cylindrical grid in \textsf{FPT} time.
We would like to insist on the fact that the proof of our main result is based on performing local changes to the proof of Kawarabayashi and Kreutzer given in the available full version of~\cite{Kawarabayashi:2015:DGT:2746539.2746586}.
In what follows we detail our results and techniques, along with the organization of the article.

In Section~\ref{section:formal_def} we give all the necessary definitions and preliminaries, and we formally state the two main contributions of this paper, namely Theorem~\ref{theorem:first_contribution_intro} and Theorem~\ref{theorem:path_sec_contribution}. As discussed above, in Section~\ref{section:finding_a_cylindrical_grid} we sketch how these two results, combined with the remainder of the original proof in~\cite{Kawarabayashi:2015:DGT:2746539.2746586}, yield the  \textsf{FPT} algorithm stated in Corollary~\ref{theorem:second_contribution}.


Similarly to the undirected case (see, for example,~\cite[Chapter 11]{Parameterized.Complexity.Theory}), the result by Johnson et al.~\cite[3.3]{Johnson.Robertson.Seymour.Thomas.01} shows that the size of a special kind of vertex separator of some set $T \subseteq V(D)$ is intrinsically connected to the directed tree-width of $D$.
Their algorithm runs a subroutine that, given a set of vertices $T$ with $|T| \leq 2k-1$, searches for a set $Z \subseteq V(D)$ with $|Z| \leq k-1$ such that every strong component of $D \setminus Z$ intersects at most half of the vertices of $T$, or decides that none exists. 
Such a set $Z$ is known as a \emph{$T$-balanced separator}~\cite{Classes.Directed.Graphs}.
If every such search is successful, the algorithm produces an arboreal decomposition of width at most $3k-2$.
If the search fails for some set $T$, then we say that $T$ is \emph{$(k-1)$-linked}~\cite{Classes.Directed.Graphs} and use it to construct a haven of order $k$ (we show how to do this construction in Lemma~\ref{claim:haven_instance}).
In Section~\ref{sec:FPT-arboreal}, we give an {\sf FPT} algorithm (Theorem~\ref{theorem:first_contribution_extended}) that, given a digraph $D$ and a parameter $k$, outputs either an arboreal decomposition of $D$ of width at most $3k-2$ or a $(k-1)$-linked set $T$ with $|T| = 2k-1$, thus improving the result by Johnson et al.~\cite{Johnson.Robertson.Seymour.Thomas.01}, since we can easily extract a haven of order $k$ from $T$ (Lemma~\ref{claim:haven_instance}), and proving our first main contribution (Theorem~\ref{theorem:first_contribution_intro}).


We acknowledge that a sketch of a proof of a similar result, with approximation factor of $5k+10$, is given in~\cite[Theorem 9.4.4]{Classes.Directed.Graphs}.
In their proof, the authors mention how to compute a weaker version of $T$-balanced separators in {\sf FPT} time with parameter $|T|$,
and the increase on the approximation factor they guarantee is a consequence of this relaxation.
For our approximation algorithm for directed tree-width, we introduce generalized versions of balanced separators and $k$-linked sets that are also used in Section~\ref{section:cyd_grid_fpt}.
Namely, we say that a set $Z$ is a \emph{$(T,r)$-balanced separator} if every strong component of $D \setminus Z$ intersects at most $r$ vertices of $T$ and that $T$ is \emph{$(k,r)$-linked} if every $(T,r)$-balanced separator has size at least $k+1$ (see Definition~\ref{def:balanced_sep_k_linked_sets}). 

In Theorem~\ref{theorem:FPT_haven} we show that the problem of finding a $(T, r)$-balanced separator of size $s$ or deciding that $T$ is $(s,r)$-linked is \textsf{FPT} with relation to the parameter $|T|$.
We refer to this problem as \textsc{Balanced Separator} and, to solve it, we make use of an algorithm by Erbacher et al.~\cite{Erbacher.Jaeger.14} for a variation of the \textsc{Multicut} problem for digraphs, named as \textsc{Multicut With Linearly Ordered Terminals} by the authors.



Next, we prove our second main contribution (Theorem~\ref{theorem:path_sec_contribution}).
For this, we need to find a bramble $\mathcal{B}$ when the second output of the algorithm for approximate arboreal decompositions (the set $T$) is obtained, and use it to find a path $P$ containing a well-linked set $A \subseteq V(P)$ of size roughly $\sqrt{2k}$.
In order to prove Theorem~\ref{theorem:path_sec_contribution}, we proceed as follows.

In Section~\ref{sec:brambles-large-dtw} we show how to construct, from a $(k-1)$-linked set $T$ with $|T| = 2k-1$, a bramble $\mathcal{B}_T$ that is easier to work with than the general case in a number of ways.
We characterize hitting sets of $\mathcal{B}_T$ by $T$-balanced separators (Lemma~\ref{lemma:BT_order_hitting_set}) and thus applying our algorithm for \textsc{Balanced Separator}, we conclude that we can decide if the order of $\mathcal{B}_T$ is at most $s$ in \textsf{FPT} time.
In fact, we prove a slightly stronger result stating that the same can be done for some particular choices of subsets (``sub-brambles'') of $\mathcal{B}_T$.
This is a considerable improvement on the running time of the naive approach to find hitting sets of brambles, which involves going through every element of the bramble.
In particular, our characterization of hitting sets of $\mathcal{B}_T$ allows us to test if given a set $X \subseteq V(D)$ is a hitting set of $\mathcal{B}_T$ in polynomial time by enumerating the strong components of $D \setminus X$.
This is an easy observation that also holds for the bramble used in the proof of the Directed Grid Theorem~\cite{Kawarabayashi:2015:DGT:2746539.2746586}.

In Section~\ref{sec:finding_P_and_A} we show how to find $P$ and $A$.
To find $P$, we iteratively grow a path until it is a hitting set of $\mathcal{B}_T$, at each time adding one vertex and testing if the current set of vertices of the growing path is a hitting set of $\mathcal{B}_T$ (Lemma~\ref{lemma:bramble_path_polytime}).
To find $A$, we produce an ordered sequence of subpaths of $P$ each being a hitting set of a ``sub-bramble'' of $\mathcal{B}_T$ of adequate order, and pick the vertices of $A$ from the vertices between those subpaths (Lemma~\ref{lemma:split_implies_well_linked}).
The key ingredient of the first procedure is the fact that we can decide if a given set of vertices $X$ is a hitting set of $\mathcal{B}_T$ in polynomial time, as a consequence of the characterization given by Lemma~\ref{lemma:BT_order_hitting_set}.
For the second procedure, we iteratively use our algorithm for \textsc{Balanced Separator} (Theorem~\ref{theorem:FPT_haven}) to test if a bramble that is formed by a particular subset of $\mathcal{B}_T$ has adequate order.
Thus, in contrast with what is done in~\cite{Kawarabayashi:2015:DGT:2746539.2746586}\footnote{Specifically, in Lemmas 4.3 and 4.4 of the full version.}, our version of the first procedure runs in polynomial time and our version of the second procedure runs in \textsf{FPT} time, assuming that $\mathcal{B}_T$ and $T$ are given in both cases.
In order to prove Theorem~\ref{theorem:path_sec_contribution}, we introduce the notion of ``$(i)$-split'' (see Definition~\ref{definition:i-splits}), and we prove several lemmas about $(i)$-splits, namely Lemma~\ref{lemma:split_implies_well_linked} and Lemma~\ref{lemma:split_iteration}.

A roadmap of the aforementioned algorithm is given in Figure~\ref{fig:automata}.
We mark by a dashed arc the steps of~\cite{Kawarabayashi:2015:DGT:2746539.2746586} which are already {\sf FPT} and do not need to be adapted.
All others arcs represent steps that we adapt in this paper.


\newsavebox{\cydgrid}
\begin{lrbox}{\cydgrid}
\begin{tikzpicture}

\def \n {8}
\def \margin {3}
\def \radiusinner {1cm}
\foreach \i in {1,...,\n}
{
  \setcounter{c}{\i};
  \def \vname {\alph{c}};

  \node[blackvertex,vertex,scale=.4] (P-\vname) at ({360/\n * (\i - 1)}:\radiusinner) {};
  \draw[-{Latex[length=2mm, width=2mm]}] ({360/\n * (\i - 1)+\margin}:\radiusinner)
  arc ({360/\n * (\i - 1)+\margin}:{360/\n * (\i)-\margin}:\radiusinner);
}

\def \radiusinner {2cm}

\foreach \i in {1,...,\n}
{
  \node[blackvertex,vertex,scale=.4] at ({360/\n * (\i - 1)}:\radiusinner) {};
  \draw[-{Latex[length=2mm, width=2mm]}] ({360/\n * (\i - 1)+\margin}:\radiusinner)
  arc ({360/\n * (\i - 1)+\margin}:{360/\n * (\i)-\margin}:\radiusinner);
}

\def \radiusinner {3cm}

\foreach \i in {1,...,\n}
{

  \node[blackvertex,vertex,scale=.4] at ({360/\n * (\i - 1)}:\radiusinner) {};
  \draw[-{Latex[length=2mm, width=2mm]}] ({360/\n * (\i - 1)+\margin}:\radiusinner)
  arc ({360/\n * (\i - 1)+\margin}:{360/\n * (\i)-\margin}:\radiusinner);
}

\def \radiusinner {4cm}

\foreach \i in {1,...,\n}
{
  \setcounter{c}{\i};
  \def \vname {\Alph{c}};

  \node[blackvertex,vertex,scale=.4] (P-\vname) at ({360/\n * (\i - 1)}:\radiusinner) {};
  \draw[-{Latex[length=2mm, width=2mm]}] ({360/\n * (\i - 1)+\margin}:\radiusinner)
  arc ({360/\n * (\i - 1)+\margin}:{360/\n * (\i)-\margin}:\radiusinner);
}

\draw[-{Latex[length=2mm, width=2mm]}] (P-a) to (P-A);
\draw[-{Latex[length=2mm, width=2mm]}] (P-B) to (P-b);
\draw[-{Latex[length=2mm, width=2mm]}] (P-c) to (P-C);
\draw[-{Latex[length=2mm, width=2mm]}] (P-D) to (P-d);
\draw[-{Latex[length=2mm, width=2mm]}] (P-e) to (P-E);
\draw[-{Latex[length=2mm, width=2mm]}] (P-F) to (P-f);
\draw[-{Latex[length=2mm, width=2mm]}] (P-g) to (P-G);
\draw[-{Latex[length=2mm, width=2mm]}] (P-H) to (P-h);
\end{tikzpicture}
\end{lrbox}

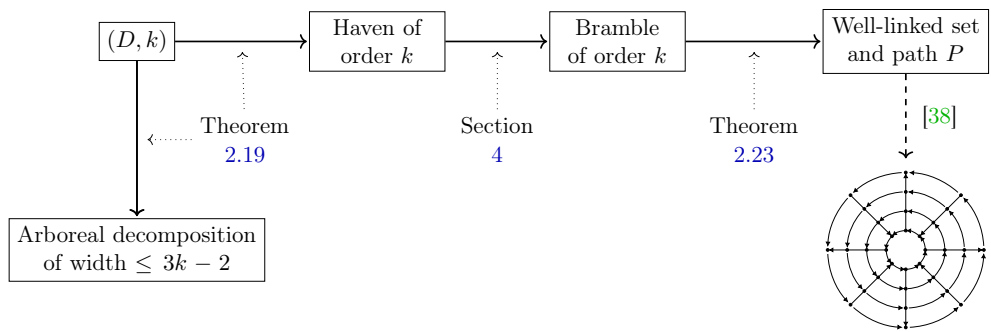
\begin{figure}[ht]
\centering
\scalebox{.79}{\begin{tikzpicture}[shorten >=1pt,->]]
\node[rectangle,draw, scale=1,text width=1cm,align=center] (P-a) at (3*0,0) {$(D,k)$};
\node[rectangle,draw,scale=1,text width=2cm,align=center] (P-b) at (4*1,0) {Haven of order $k$};
\node[rectangle,draw,scale=1,text width=2cm,align=center] (P-c) at (4*2,0) {Bramble of order $k$};
\node[rectangle,draw,scale=1,text width=2.5cm,align=center] (P-d) at (4*3.2,0) {Well-linked set and path $P$};
\node[scale=1.3,align=center] (P-e) at (4*3.2,-3.5) {\scalebox{.25}{\usebox{\cydgrid}}};
\node[rectangle,draw,scale=1,text width=4cm,align=center] (P-f) at (3*0,-3.5) {Arboreal decomposition of width $\leq 3k-2$};

\draw[thick] (P-a) -- (P-b) node[pos=.5] (P-p2) {};
\draw[thick] (P-a) -- (P-f) node[pos=.5] (P-p1) {};
\draw[thick] (P-a) -- (P-f) node[pos=.5, xshift = 1.8cm, text width=1.5cm, align=center] (P-p3) {Theorem \ref{theorem:first_contribution_intro}};
\draw[dotted] (P-p3) -- (P-p1);
\draw[dotted] (P-p3) -- (P-p2);

\draw[thick, dashed] (P-d) -- (P-e) node[pos=.5, xshift=.5cm] {\cite{Kawarabayashi:2015:DGT:2746539.2746586}};
\phantom{\draw[thick] (P-b) -- (P-c) node[pos=.5] (P-x) {};}
\phantom{\draw[thick] (P-c) -- (P-d) node[pos=.5] (P-y) {};}
\draw[thick] (P-b) -- (P-c) node[pos=.5, yshift =-1.65cm, text width = 1.5cm, align=center] (P-x1) {Section \ref{section:cyd_grid_fpt}};
\draw[thick] (P-c) -- (P-d) node[pos=.5, yshift =-1.65cm, text width = 1.5cm, align=center] (P-y1) {Theorem \ref{theorem:path_sec_contribution}};
\draw[dotted] (P-x1) -- (P-x);
\draw[dotted] (P-y1) -- (P-y);
\end{tikzpicture}%
}%
\caption{Sketch of the algorithm used in the proof of the Directed Grid Theorem~\cite{Kawarabayashi:2015:DGT:2746539.2746586}.}
\label{fig:automata}
\end{figure}

We conclude the article in Section~\ref{section:concluding_remarks} with some remarks and potential algorithmic applications of our results.


\section{Formal definitions and preliminaries}\label{section:formal_def}
In this section we give the definitions relevant to this paper, mention some known results, and present a more detailed discussion of our main contributions.

\subsection{Graphs and digraphs}
We refer the reader to~\cite{Graph.Theory} for basic background on graph theory, and recall here only some basic definitions. For a graph $G = (V,E)$, directed or not, and a set $X \subseteq V(G)$, we write $G \setminus X$ for the graph resulting from the deletion of $X$ from $G$.
If $e$ is an edge of a directed or undirected graph with \emph{endpoints} $u$ and $v$, we may refer to $e$ as $(u,v)$ and say that $e$ is \emph{incident} to $u$ and $v$.
If $e$ is an edge from $u$ to $v$ of a directed graph, we say that $e$ has \emph{tail} $u$, \emph{head} $v$, and is \emph{oriented} from $u$ to $v$. 
We also allow for loops and multiple edges.

The \emph{in-degree} (resp. \emph{out-degree}) of a vertex $v$ in a digraph $D$ is the number of edges with head (resp. tail) $v$.
The \emph{in-neighborhood} $N^-_D(v)$ of $v$ is the set $\{u \in V(D) \mid (u,v) \in E(G)\}$, and the \emph{out-neighborhood} $N^+_D(v)$ is the set $\{u \in V(D) \mid (v,u) \in E(G)\}$.
We say that $u$ is an \emph{in-neighbor} of $v$ if $u \in N^-_D(v)$ and that $u$ is an \emph{out-neighbor} of $v$ if $u \in N^+_D(v)$.

A \emph{walk} in a digraph $D$ is an alternating sequence $W$ of vertices and edges that starts and ends with a vertex, and such that for every edge $(u,v)$ in the walk, vertex $u$ (resp. vertex $v$) is the element right before (resp. right after) edge $(u,v)$ in $W$.
If the first vertex in a walk is $u$ and the last one is $v$, then we say this is a \emph{walk from $u$ to $v$}.
A \emph{path} is a digraph containing exactly a walk that contains all of its vertices and edges without repetition. 
If $P$ is a path with $V(P) = \{v_1, \ldots, v_k\}$ and $E(P) = \{(v_i, v_{i+1}) \mid i \in [k-1]\}$, we say that $v_1$ is the \emph{first} vertex of $P$, that $v_k$ is the \emph{last} vertex of $P$, and for $i \in [k-1]$ we say that $v_{i+1}$ is the \emph{sucessor in $P$} of $v_{i}$.
All paths mentioned henceforth, unless stated otherwise, are considered to be directed.

%

An \emph{orientation} of an undirected graph $G$ is a digraph $D$ obtained from $G$ by choosing an orientation for each edge $e \in E(G)$.
The undirected graph $G$ formed by ignoring the orientation of the edges of a digraph $D$ is the \emph{underlying graph} of $D$.

A digraph $D$ is \emph{strongly connected} if, for every pair of vertices $u,v \in V(D)$, there is a walk from $u$ to $v$ and a walk from $v$ to $u$ in $D$.
We say that $D$ is \emph{weakly connected} if the underlying graph of $D$ is connected.
A \emph{separator} of $D$ is a set $S \subsetneq V(D)$ such that $D \setminus S$ is not strongly connected.
If $|V(D)| \geq k+1$ and $k$ is the minimum size of a separator of $D$, we say that $D$ is \emph{$k$-strongly connected}.
A \emph{strong component} of $D$ is a maximal induced subdigraph of $D$ that is strongly connected, and a \emph{weak component} of $D$ is a maximal induced subdigraph of $D$ that is weakly connected.

For a positive integer $k$, we denote by $[k]$ the set containing every integer $i$ such that $1 \leq i \leq k$.

\subsection{Parameterized complexity}

We refer the reader to~\cite{DF13,CyganFKLMPPS15} for basic background on parameterized complexity, and we recall here only the definitions used in this article. A \emph{parameterized problem} is a language $L \subseteq \Sigma^* \times \mathbb{N}$.  For an instance $I=(x,k) \in \Sigma^* \times \mathbb{N}$, $k$ is called the \emph{parameter}.

A parameterized problem $L$ is \emph{fixed-parameter tractable} ({\sf FPT}) if there exists an algorithm $\mathcal{A}$, a computable function $f$, and a constant $c$ such that given an instance $I=(x,k)$, $\mathcal{A}$   (called an {\sf FPT} \emph{algorithm}) correctly decides whether $I \in L$ in time bounded by $f(k) \cdot |I|^c$. For instance, the \textsc{Vertex Cover} problem parameterized by the size of the solution is {\sf FPT}.

A parameterized problem $L$ is in {\sf XP} if there exists an algorithm $\mathcal{A}$ and two computable functions $f$ and $g$ such that given an instance $I=(x,k)$, $\mathcal{A}$  (called an {\sf XP} \emph{algorithm}) correctly decides whether $I \in L$ in time bounded by $f(k) \cdot |I|^{g(k)}$. For instance,  the \textsc{Clique} problem parameterized by the size of the solution is in  {\sf XP}.

Within parameterized problems, the class {\sf W}[1] may be seen as the parameterized equivalent to the class {\sf NP} of classical decision problems. Without entering into details (see~\cite{DF13,CyganFKLMPPS15} for the formal definitions), a parameterized problem being {\sf W}[1]-\emph{hard} can be seen as a strong evidence that this problem is {\sl not} {\sf FPT}.
The canonical example of {\sf W}[1]-hard problem is \textsc{Clique}  parameterized by the size of the solution.

\subsection{Arboreal decompositions and obstructions}\label{section:arboreal_decomp_obstructions}

By an \emph{arborescence} $R$ with \emph{root} $r_0$, we mean an orientation of a tree such that $R$ contains a path from $r_0$ to every other vertex of the tree.
If a vertex $v$ of $R$ has out-degree zero, we say that $v$ is a \emph{leaf} of $R$. We now define guarded sets and arboreal decompositions of digraphs. From here on, we refer to oriented edges only, unless stated otherwise. $D$ will always stand for a digraph, and $G$ for an undirected graph.
Unless stated otherwise, we define $n = |V(D)|$ and $m = |E(D)|$ when $D$ is the input digraph of some algorithm.

For $X, Y \subseteq V(D)$, an $(X,Y)$-\emph{separator} is a set of vertices $S$ such that there are no paths in $D \setminus S$ from any vertex in $X$ to any vertex in $Y$.
We make use Menger's Theorem~\cite{Menger1927} for digraphs.
\begin{theorem}[{Menger's Theorem}~\cite{Menger1927}]
\label{thm:Menger}
Let $D$ be a digraph and $X,Y \subseteq V(D)$. 
Then the minimum size of an $(X,Y)$-separator in $D$ equals the maximum number of pairwise internally vertex-disjoint paths from $X$ to $Y$ in $D$.
\end{theorem}

\begin{definition}[$Z$-guarded sets]
Let $D$ be a digraph, $Z \subseteq V(D)$, and $S \subseteq V(D) \setminus Z$.
We say that $S$ is \emph{$Z$-guarded} if there is no directed walk in $D \setminus Z$ with first and last vertices in $S$ that uses a vertex of $D \setminus (Z \cup S)$.
\end{definition}
\noindent That is, informally speaking, a set $S$ is $Z$-guarded if whenever a walk starting in $S$ leaves $S$, it is impossible to come back to $S$ without visiting a vertex in $Z$. See Figure~\ref{fig:z_guarded_set} for an illustration of a $Z$-guarded set.
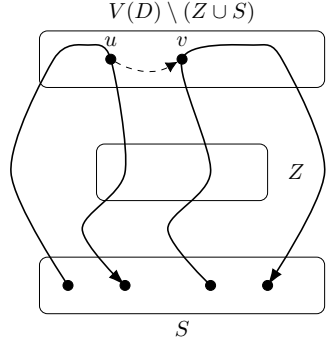
\begin{figure}[h!]
\centering
\scalebox{.75}{
\begin{tikzpicture}[scale=1]
\draw[rounded corners] (0,4) rectangle  (5,5) node [above,xshift=-2.5cm] {$V(D)\setminus (Z \cup S)$} ;
\draw[rounded corners] (0,0) rectangle (5,1) node [below,yshift=-1cm,xshift=-2.5cm] {$S$};
\draw[rounded corners] (1,2) rectangle (4,3) node [above,xshift=.5cm,yshift=-.75cm] {$Z$};
\node[blackvertex,scale=.5] (P-a) at (0.5,0.5) {};
\node (P-b) at (-0.5,2.5) {};
\node (P-c) at (0.25,4.5) {};
\node (P-d) at (0.85,4.75) {};
\node[blackvertex,label=$u$,scale=.5] (P-e) at (1.25,4.5) {};
\node (P-f) at (1.5,2.5) {};
\node (P-g) at (0.75,1.5) {};
\node[blackvertex,scale=.5] (P-h) at (1.5,0.5) {};

\draw[thick, -{Latex[length=3mm, width=2mm]}] plot[smooth] coordinates {(P-a) (P-b) (P-c) (P-d) (P-e) (P-f) (P-g) (P-h) };

\node[blackvertex,scale=.5] (P-a1) at (4,0.5) {};
\node (P-b1) at (5,2.5) {};
\node (P-c1) at (4.25,4.5) {};
\node (P-d1) at (3.75,4.75) {};
\node[blackvertex,label=$v$,scale=.5] (P-e1) at (2.5,4.5) {};
\node (P-f1) at (3,2.5) {};
\node (P-g1) at (2.25,1.5) {};
\node[blackvertex,scale=.5] (P-h1) at (3,0.5) {};

\draw[thick, -{Latex[length=3mm, width=2mm]}] plot[smooth] coordinates {(P-h1) (P-g1) (P-f1) (P-e1) (P-d1) (P-c1) (P-b1) (P-a1)};

\draw[dashed, -{Latex[length=2mm, width=2mm]}] (P-e) to [bend right =30] (P-e1);


\end{tikzpicture}%
}%
\caption{A $Z$-guarded set $S$. The dashed line indicates that there is no path from $u$ to $v$ in $V(D) \setminus (Z \cup S)$.}
\label{fig:z_guarded_set}
\end{figure}
If a set $S$ is $Z$-guarded, we may also say that $Z$ is a \emph{guard} for $S$. We remark that in~\cite{Johnson.Robertson.Seymour.Thomas.01}, the authors use the terminology of $Z$-\emph{normal} sets instead of $Z$-guarded sets.

Let $R$ be an arborescence, $r \in V(R)$, $e \in E(R)$, and $r'$ be the head of $e$.
We say that $r > e$ if there is a path from $r'$ to $r$ in $R$. 
We also say that $e \sim r$ if $r$ is the head or the tail of $e$.
To define the tree-width of directed graphs, we first need to introduce arboreal decompositions.
\begin{definition}[Arboreal decomposition]
An \emph{arboreal decomposition} $\beta$ of a digraph $D$ is a triple $(R,\mathcal{X},\mathcal{W})$ where $R$ is an arborescence, $\mathcal{X} = \{X_e : e \in E(R)\}$, $\mathcal{W} = \{W_r : r \in V(R)\}$, and $\mathcal{X},\mathcal{W}$ are collections of sets of vertices of $D$ (called \emph{bags}) such that
\begin{enumerate}
\item[\textbf{(i)}] $\mathcal{W}$ is a partition of $V(D)$ into non-empty sets, and
\item[\textbf{(ii)}] if $e \in E(R)$, then $\bigcup\{W_r : r \in V(R) \text{ and } r > e\}$ is $X_e$-guarded.
\end{enumerate}
We also say that $r$ is a \emph{leaf} of $(R,\mathcal{X,W})$ if $r$ has out-degree zero in $R$.
\end{definition}
The left hand side of Figure~\ref{fig:arboreal_decomposition} contains an example of a digraph $D$, while the right hand side shows an arboreal decomposition for it.
In the illustration of the arboreal decomposition, squares are guards $X_e$ and circles are bags of vertices $W_r$.
For example, consider the edge $e \in E(R)$ with $X_e = \{b,c\}$ from the bag $W_1$ to the bag $W_2$.
Then $\bigcup\{W_r : r \in V(R) \text{ and } r > e\} = V(D) \setminus \{a\}$ and, by item \textbf{(ii)} described above, this set must be $\{b,c\}$-guarded since $X_e = \{b,c\}$.
In other words, there cannot be a walk in $D \setminus \{b,c\}$ starting and ending in $V(D) \setminus \{a\}$ using a vertex of $\{a\}$.
This is true in $D$ since every path reaching $\{a\}$ from the remaining of the graph must do so through vertices $b$ or $c$.
The reader is encouraged to verify the same properties for the other guards in the decomposition.
\begin{figure}[ht]
\centering
\begin{subfigure}{.35\textwidth}
\scalebox{.85}{\begin{tikzpicture}

\foreach \x/\y/\name/\lpos in {
  0/5/a/90,
  -1.5/4/b/90, 1.5/4/c/90}
  \node[blackvertex, label=\lpos:{$\name$},scale=.5] (P-\name) at (\x,1.2*\y) {$\name$};

\foreach \x/\y/\name/\lpos in {
  -2.25/3/d/120, -0.75/3/e/60, 0.75/3/f/120, 2.25/3/g/60}
  \node[blackvertex, label={[label distance= -.1cm]\lpos:{$\name$}},scale=.5] (P-\name) at (\x,1.2*\y) {$\name$};

\foreach \x/\y in {
a/b, a/c, b/c, b/d, b/e, c/b, c/f, c/g, d/b, e/b, f/c, g/c, d/e, f/g}
  \draw[edge, {Latex[length=2mm, width=2mm]}-{Latex[length=2mm, width=2mm]}, line width = 1, shorten >= .1cm, shorten <= .1cm] (P-\x) to (P-\y);

\end{tikzpicture}%
}
\end{subfigure}\hspace{1cm}
\begin{subfigure}{.35\textwidth}
\scalebox{.75}{\begin{tikzpicture}
  \foreach \x/\y/\name/\idn/\lbl/\lblpos in {
  0/5/w1/a/{W_1}/180,
  0/2.5/w2/{b,c}/{W_2}/270,
  -2.5/2/w3/{d,e}/{W_3}/90, 2.5/2/w4/{f,g}/{W_4}/90}
  \node[vertex, text width=.75cm, align=center, label, label=\lblpos:{\large$\lbl$}] (P-\name) at (\x,\y) {\idn};
\draw[edge, line width = 1] (P-w1) to  node[midway, rectangle, fill=white,draw] {$b,c$}  (P-w2);
\draw[edge, line width = 1] (P-w2) to  node[midway, rectangle, fill=white,draw] {$b$}  (P-w3);
\draw[edge, line width = 1] (P-w2) to  node[midway, rectangle, fill=white,draw] {$c$}  (P-w4);
\end{tikzpicture}}
\end{subfigure}
\caption{A digraph $D$ and an arboreal decomposition of $D$ of width two. A bidirectional edge is used to represent a pair of edges in both directions.}
\label{fig:arboreal_decomposition}
\end{figure}
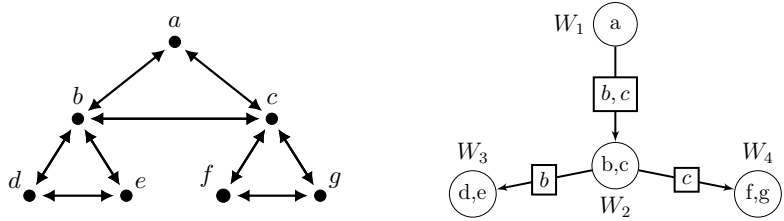
\begin{definition}[Nice arboreal decompositions]\label{def:nice_arboreal_decompositions}
We say that an arboreal decomposition $(R,\mathcal{X},\mathcal{W})$ of a digraph $D$ is \emph{nice} if
\begin{enumerate}
  \item[\textbf{(iii)}] for every $e \in E(R)$, $\bigcup\{W_r : r \in V(R), r > e\}$ induces a strong component of $D \setminus X_e$, and
  \item[\textbf{(iv)}] if $r \in V(R)$ and $r_1, \ldots, r_{\ell}$ are the out-neighbors of $r$ in $R$, then
   \[\left(\bigcup_{1 \leq i \leq \ell}W_{r_i}\right) \cap \left(\bigcup_{e \sim r}X_e\right) = \emptyset.\]
\end{enumerate}
\end{definition}
\begin{definition}[Directed tree-width]
Let $(R,\mathcal{X},\mathcal{W})$ be an arboreal decomposition of a digraph $D$.
For a vertex $r \in V(R)$, we denote by $\wdth(r)$ the size of the set $W_r \cup (\bigcup_{e \thicksim r}X_e)$.
The \emph{width} of $(R,\mathcal{X},\mathcal{W})$ is the least integer $k$ such that, for all $r \in V(R)$, $\wdth(r) \leq k+1$.
The \emph{directed tree-width}  of $D$, denoted by $\dtw(D)$, is the least integer $k$ such that $D$ has an arboreal decomposition of width $k$.
\end{definition}
\noindent We remark that DAGs have directed tree-width zero.

If $G$ is an undirected graph and $D$ the digraph obtained from $G$ by replacing every edge of $G$ with two directed edges in opposite directions then, as shown by Johnson et al.~\cite{Johnson.Robertson.Seymour.Thomas.01}, the tree-width of $G$ is equal to the directed tree-width of $D$.
Thus, deciding if a digraph $D$ has directed tree-width at most $k$, for a given integer $k$, is {\sf NP}-complete since deciding if the tree-width of an undirected graph is at most $k$ is an {\sf NP}-complete problem~\cite{doi:10.1137/0608024}.


We now formally define cylindrical grids,  butterfly contractions, butterfly minors, and some blocking structures for large directed tree-width.

\begin{definition}[Cylindrical grid]
A \emph{cylindrical grid of order $k$} is a digraph formed by the union of $k$ disjoint cycles $C_1, \ldots, C_k$ and $2k$ disjoint paths $P_1$, $P_2$, $\ldots$, $P_{2k}$ where
\begin{enumerate}
\item for $i \in [k], V(C_i) = \{v_{i,1}, v_{i,2}, \ldots, v_{i,2k}\}$ and $E(C_i) = \{(v_{i,j},v_{i, j+1} \mid j \in [2k-1])\} \cup \{(v_{i, 2k}, v_{i,1})\}$,
\item for $i \in \{1, 3, \ldots, 2k-1\}$, $E(P_i) = \{(v_{1, i},v_{2,i}), (v_{2,i},v_{3,i}), \ldots, (v_{k-1,i},v_{k,i})\}$, and
\item for $i \in \{2,4, \ldots, 2k\}$, $E(P_i) = \{(v_{k, i},v_{k-1,i}), (v_{k-1, i},v_{k-2,i}), \ldots, (v_{2,i},v_{1,i})\}$.
\end{enumerate}
\end{definition}
\noindent In other words, path $P_i$ is oriented from the first circle to the last one if $i$ is odd, and the other way around if $i$ is even.
Furthermore, every vertex of a cylindrical grid occurs in the intersection of a path and a cycle.
See Figure~\ref{fig:cylindrical_grid_intro} for an example of a cylindrical grid of order $k=4$.

\begin{definition}[Butterfly contraction and butterfly minors]
Let $D$ be a digraph.
An edge $e$ from $u$ to $v$ of $D$ is \emph{butterfly contractible} if $e$ is the only outgoing edge of $u$ or the only incoming edge of $v$.
By \emph{butterfly contracting} $e$ in $D$, we obtain a digraph $D'$ with vertex set $V(D') = V(D) \setminus  \{u,v\} \cup \{x_{u,v}\}$, where $x_{u,v}$ is a new vertex, and  $E(D') = E(D) \setminus \{e\}$.
Every incidence of an edge $f \in E(D')$ to $u$ or $v$ in $D$ becomes an incidence to $x_{u,v}$ in $D'$.
If $D'$ is generated from a subgraph of $D$ by a series a butterfly contractions, we say that $D'$ is a \emph{butterfly minor} of $D$.
\end{definition}
\noindent Notice that, in the above definition, the newly introduced vertex $x_{u,v}$ has in $D'$ the same neighbors of $u$ and $v$ in $D$.
It is not hard to see that butterfly contractions cannot generate any new paths, and that there is no such guarantee if no restrictions are imposed on which edges of a digraph can be contracted.
See Figure~\ref{fig:butterfly_contractions} for an example of this.

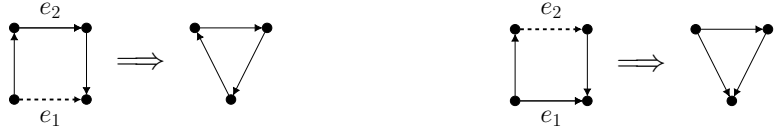
\begin{figure}[h!]
\centering
\begin{subfigure}{.4\textwidth}
\centering
\scalebox{.475}{\begin{tikzpicture}[scale=2]
\foreach \x/\y/\name in {
  0/0/a, 0/1/b, 1/1/c, 1/0/d} {
  \node[blackvertex, scale=0.75] (P-\name) at (\x,\y) {};
  }
\foreach \x/\y/\name in {
  0.5/0/ad, 0/1/b, 1/1/c} {
  \node[blackvertex, scale=0.75] (Q-\name) at (\x+2.5,\y) {};
  }

\foreach \x/\y in {
a/b, b/c, c/d}{
  \draw[-{Latex[length=2mm, width=2mm]}] (P-\x) -- (P-\y);
}
\foreach \x/\y in {
ad/b, b/c, c/ad}{
  \draw[-{Latex[length=2mm, width=2mm]}] (Q-\x) -- (Q-\y);
}
  \draw[-{Latex[length=2mm, width=2mm]}, line width=1.5, dashed] (P-a) -- (P-d) node [midway, label=-90:{\huge$e_1$}] {};
  \draw[-{Latex[length=2mm, width=2mm]}] (P-b) --  (P-c) node [midway, label=90:{\huge$e_2$}] {};

\node (P-i) at (1.75,0.5) {\Huge$\implies$};
\end{tikzpicture}}%
\end{subfigure}\hspace{1cm}
\begin{subfigure}{.4\textwidth}
\centering
\scalebox{.475}{\begin{tikzpicture}[scale=2]
\foreach \x/\y/\name in {
  0/0/a, 0/1/b, 1/1/c, 1/0/d} {
  \node[blackvertex, scale=0.75] (P-\name) at (\x,\y) {};
  }
\foreach \x/\y/\name in {
  0.5/0/ad, 0/1/b, 1/1/c} {
  \node[blackvertex, scale=0.75] (Q-\name) at (\x+2.5,\y) {};
  }

\foreach \x/\y in {
a/b, c/d}{
  \draw[-{Latex[length=2mm, width=2mm]}] (P-\x) --  (P-\y);
}
\draw[-{Latex[length=2mm, width=2mm]}] (P-a) -- (P-d) node [midway, label=-90:{\huge$e_1$}] {};
\draw[-{Latex[length=2mm, width=2mm]},line width =1.5, dashed] (P-b) --  (P-c) node [midway, label=90:{\huge$e_2$}] {};
\foreach \x/\y in {
b/ad, b/c, c/ad}{
  \draw[-{Latex[length=2mm, width=2mm]}] (Q-\x) -- (Q-\y);
}
  \draw[-{Latex[length=2mm, width=2mm]}] (P-a) -- (P-d);

\node (P-i) at (1.75,0.5) {\Huge$\implies$};
\end{tikzpicture}}%
\end{subfigure}
\caption{Butterfly contractions preserve separations. In each example the dashed edge is contracted to generate the digraph on the right. Edge $e_1$ is \textsl{not} butterfly contractible.}
\label{fig:butterfly_contractions}
\end{figure}

\begin{definition}[Well-linked sets]\label{def:well-linked-sets-in-digraphs}
Let $D$ be a digraph and $A \subseteq V(D)$.
We say that $A$ is \emph{well-linked} in $D$ if, for all disjoint $X,Y \subseteq A$ with $|X| = |Y|$, there are $|X|$ vertex-disjoint paths from $X$ to $Y$ in $D$.
The \emph{order} of a well-linked set $A$ is $|A|$.
We denote by $\wlink(D)$ the size of a largest well-linked set in $D$.
\end{definition}
\begin{definition}[Havens in digraphs]\label{def:havens_in_digraphs}
Let $D$ be a digraph. A \emph{haven of order $k$} in $D$ is a function $\beta$ assigning to every set $Z \subseteq V(D)$, with $|Z| \leq k-1$, the vertex set of a strong component of $D \setminus Z$ in such way that if $Z' \subseteq Z \subseteq V(D)$ then $\beta(Z) \subseteq \beta(Z')$. The \emph{haven number} of a digraph $D$, denoted by $\hn(D)$, is the maximum $k$ such that $D$ admits a haven of order $k$.
\end{definition}
\noindent A $k$-strongly connected digraph, for example, admits a haven of order $k$: it suffices to choose $\beta(Z) = V(D) \setminus Z$ for any $Z \subseteq V(D)$ with $|Z| \leq k-1$.
Figure~\ref{figure:haven_property} illustrates the defining property of havens.
\begin{figure}[h]
\centering
\scalebox{.75}{
\begin{tikzpicture}[scale=1]
  \draw (0,0) circle (1.5cm) node[anchor=center,yshift=1.75cm] {\large$Z$};
  \draw (0,0) circle (.5cm) node[anchor=center] {\large$Z'$};

  \draw (5.5cm,0) circle (1.5cm) node[anchor=center,yshift=1.75cm] {\large$\beta(Z')$};
  \draw (5.5cm,0) circle (.5cm) node[anchor=center] {\large$\beta(Z)$};

  \draw[dashed] (0,1.5cm) -- (5.5cm, .5cm);
  \draw[dashed] (0,-1.5cm) -- (5.5cm, -.5cm);
  \draw[dashed] (0,.5cm) -- (5.5cm, 1.5cm);
  \draw[dashed] (0,-.5cm) -- (5.5cm, -1.5cm);
\end{tikzpicture}%
}%
\caption{Illustration of the haven property.}
\label{figure:haven_property}
\end{figure}
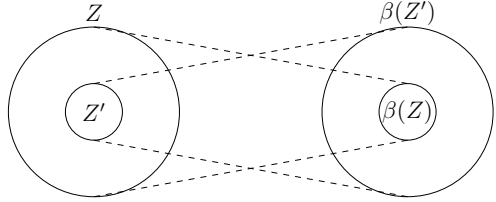

\begin{definition}[Brambles in digraphs]
A \emph{bramble} $\mathcal{B} = \{B_1, \ldots, B_\ell\}$ in a digraph $D$ is a family of strongly connected subgraphs of $D$ such that if $\{B, B'\} \subseteq \mathcal{B}$ then $V(B) \cap V(B') \neq \emptyset$ or there are edges in $D$ from $V(B)$ to $V(B')$ and from $V(B')$ to $V(B)$.
A \emph{hitting set} of a bramble $\mathcal{B}$ is a set $C \subseteq V(D)$ such that $C \cap V(B) \neq \emptyset$ for all $B \in \mathcal{B}$. The \emph{order}  of a bramble $\mathcal{B}$, denoted by $\order(\mathcal{B})$,  is the minimum size of a hitting set of $\mathcal{B}$. The \emph{bramble number} of a digraph $D$, denoted by $\bn(D)$, is the the maximum $k$ such that $D$ admits a bramble of order $k$.
\end{definition}

There is a direct relation between the haven number and the tree-width of undirected graphs.
A haven in an undirected graph is defined similarly: the function $\beta$ retains all its properties, but mapping sets of at most $k-1$ vertices to components of the graph resulting from the deletion of those vertices.

\begin{proposition}[Seymour and Thomas~\cite{Seymour.Thomas.93}]
Let $G$ be an undirected graph and $k \geq 1$ be an integer.
Then $G$ has a haven of order $k$ if and only if its tree-width is at least $k-1$.
\end{proposition}
\noindent For digraphs, only one implication of the previous result is known to be true.
\begin{proposition}[Johnson et al.~\cite{Johnson.Robertson.Seymour.Thomas.01}]
\label{theorem:directed_haven_treewidth}
Let $D$ be a digraph and $k$ be a non-negative integer. If $D$ has a haven of order $k$, then $\dtw(D) \geq k-1$.
\end{proposition}
\noindent For the reverse direction of Proposition~\ref{theorem:directed_haven_treewidth}, only an approximate version is known.
\begin{proposition}[Johnson et al.~\cite{Johnson.Robertson.Seymour.Thomas.01}]
\label{theorem:directed_tree_width}
Let $D$ be a digraph and $k$ be a positive integer.
If $\dtw(D) \geq 3k-1$ then $D$ admits a haven of order $k$.
\end{proposition}
Finally, the following two lemmas show that brambles of large order and large well-linked sets are obstructions to small directed tree-width.
The proof of the first lemma can be done by converting brambles into havens and back.
For the second lemma, it is sufficient to show that any minimum hitting set of a bramble of order $k$ is well-linked and  to extract a bramble of order $k$ from a well-linked set of order $4k+1$.
The proofs are simple and can be found, for example, in~\cite[Chapter 6]{Quantitative.Graph.Theory}.

\begin{lemma}
\label{lemma:haven_bramble_relation}
Let $D$ be a digraph.
Then $\bn(D) \leq \hn(D) \leq 2\bn(D)$.
\end{lemma}
\begin{lemma}
Let $D$ be a digraph.
Then $\bn(D) \leq \wlink(D) \leq 4\bn(D)$.
\end{lemma}
The proof of Proposition~\ref{theorem:directed_tree_width} given in~\cite{Johnson.Robertson.Seymour.Thomas.01} yields an {\sf XP} algorithm that correctly states that $D$ has a haven of order $k$ or produces an arboreal decomposition of $D$ of width at most $3k-2$.
Furthermore, although not explicitly mentioned in the paper, this algorithm actually produces a nice (as in Definition~\ref{def:nice_arboreal_decompositions}) arboreal decomposition for $D$, and can be used as a procedure that, given a digraph $D'$ such that $\dtw(D') \leq k-2$, generates a nice arboreal decomposition for $D'$ of width at most $3k-2$.
At each iteration, the algorithm tests whether the strong components intersecting a given set $T \subseteq V(D)$ with $|T| \leq 2k-1$ can be separated into parts containing at most a small portion of $T$.
Namely, the algorithm tests whether there is a set $Z \subseteq V(D)$ with $|Z| \leq k-1$ such that every strong component of $D \setminus Z$ contains at most $k-1$ vertices of $T \setminus Z$.
Such a set $Z$ is known as a \emph{balanced separator}. 
In this paper we consider a generalization of such sets where we can choose how many vertices of $T$ each strong component of $D \setminus Z$ can have.
\begin{definition}[$(T,r)$-balanced separators and $(k,r)$-linked sets]\label{def:balanced_sep_k_linked_sets}
Let $D$ be a digraph, $T \subseteq V(D)$, and $r$ be a positive integer.
A \emph{$(T,r)$-balanced separator} is a set of vertices $Z \subseteq V(D)$ such that every strong component of $D \setminus Z$ contains at most $r$ vertices of $T$.
If the minimum size of a $(T, r)$-balanced separator is at least $k+1$, we say that $T$ is \emph{$(k,r)$-linked}.
\end{definition}
\noindent If $r = \lfloor |T|/2 \rfloor$, $(T, r)$-balanced separators are exactly $T$-balanced separators in the classical sense as defined, for instance, in~\cite[Chapter 9]{Classes.Directed.Graphs}.
If $D$ admits a $(T,r)$-\balsep $Z$, we know that we can split $T \setminus Z$ into small strongly connected parts which are guarded by $Z$.
See Figure~\ref{figure:balanced_separator} for two examples of $(T,r)$-\balsep{s}. 
\noindent A DAG, for instance, admits a $(T,1)$-\balsep (the empty set) for any $T \subseteq V(D)$ since every strong component of a DAG is formed by a single vertex.

\begin{figure}
\centering
\begin{subfigure}{.35\textwidth}
\scalebox{.56}{\begin{tikzpicture}

\foreach \x/\y/\name in {
  0/1.5/v1, 1/1.5/v2,
  0/0.5/v3, 1/0.5/v4,
  0/-0.5/v5, 1/-0.5/v6,
  0/-1.5/v7,
  2/0.5/v9, 2/-0.5/v10}
  \node[blackvertex, scale=.75] (P-\name) at (2*\x,1.5*\y) {};
\foreach \x/\y/\name in {
  0/1.5/v1, 1/1.5/v2,
  0/0.5/v3}
  \node[blackvertex, scale=.75] (P-\name) at (2*\x,1.5*\y) {};

\foreach \x/\y/\name in {
  0/-0.5/v5, 1/-0.5/v6,
  0/-1.5/v7}
  \node[vertex, scale=.75] (P-\name) at (2*\x,1.5*\y) {};

\node (P-a) at ($(P-v1) + (-.25, .25)$) {};
\node (P-b) at ($(2,-2.25) + (.25, -.25)$) {};
\node[draw, rounded corners, thick, fit= (P-a) (P-b), label = {\huge $T_1$}] {};
\node (P-a) at ($(P-v4) + (-.25, .25)$) {};
\node (P-b) at ($(P-v9) + (.25, -.25)$) {};
\node[draw, rounded corners, dashed, thick, fit= (P-a) (P-b), label=0:{\huge $Z$},pattern=north east lines,pattern color=black!50] {};

\foreach \x/\y in {
  v1/v3, v3/v5, v5/v7,
  v6/v4, v4/v2,
  v2/v1, v3/v4, v6/v5,
  v3/v2, v7/v6,
  v9/v2, v6/v10,
  v10/v9}
  \draw[-{Latex[length=2mm, width=2mm]}, line width = 1, shorten >= .1cm, shorten <= .1cm] (P-\x) -- (P-\y);
\end{tikzpicture}}%
\end{subfigure}\hspace{1cm}
\begin{subfigure}{.35\textwidth}
\scalebox{.8}{\begin{tikzpicture}

\foreach \x/\y/\name/\lpos in {
  0/0/v_1/90,
  -1/-1/v_2/180, 1./-1/v_3/0}
  \node[blackvertex, label=\lpos:{\large$\name$},scale=.5] (P-\name) at (\x,1.2*\y) {$\name$};

\node[blackvertex, scale= .5] (P-x12) at ($(P-v_1) + (-2.5,0)$) {};
\node[blackvertex, scale= .5] (P-x13) at ($(P-v_1) + (2.5,0)$) {};
\node[blackvertex, scale= .5] (P-x23) at ($(P-v_1) + (0,-2.5)$) {};
\node[draw, rectangle, label={\Large $T_2$}, text width = 3.3cm, text height = 2.3cm, dashed] (P-T) at (0,-.6) {};

\draw[-{Latex[length=2mm, width=2mm]}, line width = 1, shorten >= .1cm, shorten <= .1cm] (P-v_1) -- (P-x12);
\draw[-{Latex[length=2mm, width=2mm]}, line width = 1, shorten >= .1cm, shorten <= .1cm] (P-x12) -- (P-v_2);

\draw[-{Latex[length=2mm, width=2mm]}, line width = 1, shorten >= .1cm, shorten <= .1cm] (P-v_3) -- (P-x13);
\draw[-{Latex[length=2mm, width=2mm]}, line width = 1, shorten >= .1cm, shorten <= .1cm] (P-x13) -- (P-v_1);

\draw[-{Latex[length=2mm, width=2mm]}, line width = 1, shorten >= .1cm, shorten <= .1cm] (P-v_2) -- (P-x23);
\draw[-{Latex[length=2mm, width=2mm]}, line width = 1, shorten >= .1cm, shorten <= .1cm] (P-x23) -- (P-v_3);


\foreach \x/\y in {
v_1/v_2, v_2/v_3, v_3/v_1}
  \draw[edge, -{Latex[length=2mm, width=2mm]}, line width = 1, shorten >= .1cm, shorten <= .1cm] (P-\x) to (P-\y);

\end{tikzpicture}}%
\end{subfigure}
\caption{Examples of \balsep{s}. On the left, $Z$ is a $(T_1, 3)$-balanced separator, and $T_1$ is $(3, 3)$-linked. On the right, each vertex $v_i$ with $i \in [3]$ constitutes a $(T_2, 1)$-\balsep.}
\label{figure:balanced_separator}
\end{figure}
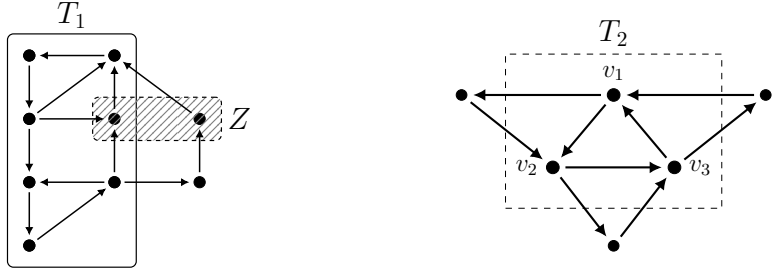

Deciding whether a digraph $D$ admits a $(T, k-1)$-balanced separator is a key ingredient for the algorithm given by Johnson et al.~\cite{Johnson.Robertson.Seymour.Thomas.01}.
Moreover, the cost of this procedure has the largest impact on the running time of their algorithm: it is the only step which is (originally) done in {\sf XP} time, while the remaining parts of the algorithm can be done in polynomial time.
In Section~\ref{section:FPT_arb_decomp}, we use of a variation of the \textsc{Multicut} problem introduced in~\cite{Erbacher.Jaeger.14} to show how to find $(T, r)$-balanced separators in {\sf FPT} time with parameter $|T|$, if any exists with size bounded from above by an integer $s$ with $s \leq |T| - 1$.
In our first main contribution, we use this result to improve on the algorithm for arboreal decompositions given in~\cite{Johnson.Robertson.Seymour.Thomas.01}.
Namely, we prove the following.
\begin{restatable}{theorem}{firstmainklinked}
\label{theorem:first_contribution_extended}
Let $D$ be a digraph and $k$ be a non-negative integer.
There is an algorithm running in time $\RunTimeAlg$ that either produces a nice arboreal decomposition of $D$ of width at most $3k-2$ or outputs a $(k-1, k-1)$-linked set $T$ with $T = 2k-1$. 
\end{restatable}

It is also not hard to see how to use $(k,r)$-linked sets to construct havens.
The following lemma is a generalization of a result shown as part of the proof of~\cite[3.3]{Johnson.Robertson.Seymour.Thomas.01}.
\begin{lemma}\label{claim:haven_instance}
Let $D$ be a graph, $T \subseteq V(D)$ with $|T| = s$, and $r \geq \lfloor s/2 \rfloor$.
If $T$ is $(k, r)$-linked then $D$ admits a haven of order $k+1$.
\end{lemma}
\begin{proof}
By hypothesis, it holds that, for every set $Z \subseteq V(D)$ with $|Z| \leq k$, there is a strong component $C$ of $D \setminus Z$ such that $|V(C) \cap T| \geq r+1$.
Let $\beta(Z) = V(C)$.
We claim that $\beta$ is a haven of order $k+1$ in $D$.
It suffices to show that if $Z' \subseteq Z$, then $\beta(Z) \subseteq \beta(Z')$.
Notice that $\beta(Z)$ induces a strongly connected subgraph of $D$ and is disjoint from $Z'$, since it is disjoint from $Z$, and thus all paths in the graph induced by $\beta(Z)$ are in $D \setminus Z'$.
Furthermore, since $|T| = s$ and $r \geq \lfloor s/2 \rfloor$, we have $\beta(Z) \cap \beta(Z') \neq \emptyset$ and the result follows as $\beta(Z')$ is a strong component of $D \setminus Z'$, which is a supergraph of $D \setminus Z$, and thus it must contain completely the strongly connected subgraph induced by $\beta(Z)$.
\end{proof}
Applying this lemma on a $(k-1, k-1)$-linked set $T$ with $|T| = 2k-1$ we obtain a haven of order $k$ and therefore we can write Theorem~\ref{theorem:first_contribution_extended} with havens instead of $(k, r)$-linked sets, as done by Johnson et al.~\cite[3.3]{Johnson.Robertson.Seymour.Thomas.01}, with the guarantee that the procedure runs in \textsf{FPT} time.
\begin{restatable}[First main contribution]{theorem}{firstmain}
\label{theorem:first_contribution_intro}
Let $D$ be a digraph and $k$ be a non-negative integer.
There is an algorithm running in time $\RunTimeAlg$ that correctly states that $D$ admits a haven of order $k$ or produces an arboreal decomposition of $D$ of width at most $3k -2$.
\end{restatable}

Next, we discuss some of the steps in the proof of the Directed Grid Theorem.
\subsection{Brambles and the Directed Grid Theorem}\label{section:brambles_DGT}
The Directed Grid Theorem is as stated below.
\begin{theorem}[Kawarabayashi and Kreutzer~\cite{Kawarabayashi:2015:DGT:2746539.2746586}]\label{theorem:directed_grid_theorem}
There is a function $f: \mathbb{N} \to \mathbb{N}$ such that given any directed graph and any fixed constant $k$, in polynomial time, we can obtain either
\begin{enumerate}
  \item an arboreal decomposition of $D$ of width at most $f(k)$, or
  \item a cylindrical grid of order $k$ as a butterfly minor of $D$.
\end{enumerate}
\end{theorem}
\noindent The  proof of the Directed Grid Theorem~\cite{Kawarabayashi:2015:DGT:2746539.2746586} starts by asking if a digraph $D$ satisfies $\dtw(D) \leq f(k)$, for some integer $k$.
By Theorem~\ref{theorem:first_contribution_intro}, an approximate answer to this question can be computed in {\sf FPT} time with parameter $k \geq 0$.
If a haven is obtained, the next step uses it to construct a bramble of large order.
In order to justify our following results, we now discuss how to construct brambles from havens.

Finding a hitting set of minimum size of a bramble $\mathcal{B}$ is not an easy task.
In general, in order to check whether a given set $X$ is a hitting set of $\mathcal{B}$, the naive approach would be to go through all the elements of $\mathcal{B}$ and verify that $X$ intersects each of them.
Since a bramble $\mathcal{B}$ may contain $\Omega(2^n)$ elements, independently of its order,
this procedure is not efficient. 
For instance, consider the digraph $D$ shown in Figure~\ref{figure:bramble_example}, which has vertex set $\{v_0, v_1, \ldots, v_n\}$ and edge set $\{(v_0, v_i) \cup (v_i, v_0) \mid i \in [n] \}$.
The set $\mathcal{B} = \{D[X] \mid X \subseteq V(D) \text{ and }  v_0 \in X\}$ is easily seen to be a bramble in $D$ of order one and size $2^{|V(D)|-1}$ since there is an edge in $D$ from every vertex in $V(D) \setminus \{v_0\}$ to $v_0$ and vice-versa.
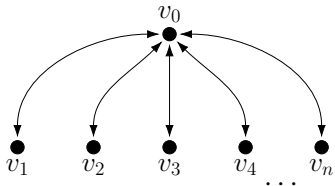
\begin{figure}[h!]
\centering
\scalebox{.5}{\begin{tikzpicture}
\node[blackvertex,label=above:{\huge $v_0$}] (P-u) at (2*3,0) {};
\foreach \x in {1, ..., 4}{
  \node[blackvertex,label=below:{\huge $v_{\x}$}] (P-v\x) at (2*\x,-3) {};
  %
}
\node[blackvertex,label=below:{\huge $v_{n}$}] (P-vn) at (2*5,-3) {};
\node[scale=2] (P-c) at (2*4.5,-4) {$\cdots$};
\draw[{Latex[length=3mm, width=2mm]}-{Latex[length=3mm, width=2mm]}, shorten >= .1cm, shorten <= .1cm, out = 180, in = 90] (P-u) to  (P-v1);
\draw[{Latex[length=3mm, width=2mm]}-{Latex[length=3mm, width=2mm]}, shorten >= .1cm, shorten <= .1cm, out = 225, in = 90] (P-u) to  (P-v2);

\draw[{Latex[length=3mm, width=2mm]}-{Latex[length=3mm, width=2mm]}, shorten >= .1cm, shorten <= .1cm] (P-u) to (P-v3);

\draw[{Latex[length=3mm, width=2mm]}-{Latex[length=3mm, width=2mm]}, shorten >= .1cm, shorten <= .1cm, out = 315, in = 90] (P-u) to  (P-v4);
\draw[{Latex[length=3mm, width=2mm]}-{Latex[length=3mm, width=2mm]}, shorten >= .1cm, shorten <= .1cm, out = 0, in = 90] (P-u) to  (P-vn);
\end{tikzpicture}}%
\caption{Example of a digraph $D$ having a bramble of order one and size $2^{|V(D)|-1}$. Here a bidirectional edge is used to represent a pair of edges in both directions.}
\label{figure:bramble_example}
\end{figure}
However, when $\mathcal{B}$ is the bramble obtained by a construction used in a proof of Lemma~\ref{lemma:haven_bramble_relation}, which we present below, then $|\mathcal{B}| = n^{\Ocal(k)}$ and thus in this case we can find hitting sets of $\mathcal{B}$ of size $k$ in \textsf{XP} time, and decide whether a given set $X \subseteq V(D)$ is a hitting set of $\mathcal{B}$ in \textsf{XP} time.

Lemma~\ref{lemma:haven_bramble_relation} implies that if $D$ is a digraph admitting a haven of order $k+1$, then $D$ contains a bramble of order at least $\lceil(k+1)/2 \rceil = \lfloor k/2\rfloor + 1$.
In fact, given such a haven, it is easy to construct the claimed bramble, as we proceed to explain.
Namely, given a haven $\beta$ of order $k+1$ in $D$, we define
$\mathcal{B} = \{D[\beta(Z)] \mid Z \subseteq V(D) \text{ and } |Z| \leq \lfloor k/2 \rfloor\}$. Note that, since $\beta$ is a haven, the elements of $\mathcal{B}$ are strongly connected subgraphs of $D$. We claim that any two elements of $\mathcal{B}$ intersect.
Indeed, let $B,B' \in \mathcal{B}$ and let $Z,Z' \subseteq V(D)$ such that $\beta(Z) = V(B)$ and $\beta(Z') = V(B')$.
Since $|Z|\leq \lfloor k/2 \rfloor$ and $|Z'|\leq \lfloor k/2 \rfloor$, we have that $|Z \cup Z'| \leq k$, and since $\beta$ is a haven of order $k+1$, it follows that $\beta (Z \cup Z') \subseteq \beta(Z) \cap \beta(Z') = V(B) \cap V(B')$ and therefore, in particular, $V(B) \cap V(B') \neq \emptyset$. Finally, let us argue about the order of $\mathcal{B}$. Consider an arbitrary vertex set $X \subseteq V(D)$ with $|X| \leq \lfloor k/2 \rfloor$. Since $\beta$ is a haven or order $k+1 \geq \lfloor k/2 \rfloor$, there is a bramble element $\beta(X) \in \mathcal{B}$ with $V(\beta(X)) \cap X = \emptyset$, and thus $\ord(\mathcal{B}) \geq \lfloor k/2\rfloor + 1$, as we wanted to prove.
Moreover, since there is one element in $\mathcal{B}$ for each $Z \subseteq V(D)$ with $|Z| \leq \lfloor k/2 \rfloor$, we conclude that $|\mathcal{B}| = n^{\Ocal(k)}$.
In~\cite{Kawarabayashi:2015:DGT:2746539.2746586}, the authors show how to obtain, from a bramble $\mathcal{B}$ of order $k(k+2)$, a path $P$ that is a hitting set of $\mathcal{B}$ containing a well-linked set $A$ of size $k$.
\begin{proposition}[Kawarabayashi and Kreutzer~{\cite[Lemma 4.3 of the full version]{Kawarabayashi:2015:DGT:2746539.2746586}}]
\label{lemma:directed_grid_path_hitting_set}
Let $D$ be a digraph and $\mathcal{B}$ be a bramble in $D$.
Then there is a path $P$ intersecting every $B \in \mathcal{B}$.
\end{proposition}

\begin{restatable}[Kawarabayashi and Kreutzer~{\cite[Lemma 4.4 of the full version]{Kawarabayashi:2015:DGT:2746539.2746586}}]{proposition}{LemmaPathSplittingDGT}
\label{lemma:directed_grid_bramble_path}
Let $D$ be a digraph, $\mathcal{B}$ be a bramble of order $k(k+2)$ in $D$, and $P = P(\mathcal{B})$ be a path intersecting every $B \in \mathcal{B}$. Then there is a set $A \subseteq V(P)$ of size $k$ which is well-linked.
\end{restatable}
Although the statements of the previous two propositions in~\cite{Kawarabayashi:2015:DGT:2746539.2746586} are not algorithmic, algorithms for both results can be extracted from their constructive proofs.
However, the naive approach to decide if a set $X \subseteq V(D)$ is a hitting set of a bramble $\mathcal{B}$ is to check if $V(B) \cap X \neq \emptyset$ for each $B \in \mathcal{B}$.
Thus the running time of the algorithms yielded by the proofs of Propositions~\ref{lemma:directed_grid_path_hitting_set} and~\ref{lemma:directed_grid_bramble_path} is influenced by the size of the bramble given as input.
Although in general this is not efficient since, as discussed above, a bramble can have size $\Omega(2^n)$ even if it has small order, in the particular case where $\mathcal{B}$ is the bramble constructed from havens as presented above, those constructions yield \textsf{XP} algorithms with parameter $k$ since $|\mathcal{B}| = n^{\Ocal(k)}$.


In Section~\ref{section:cyd_grid_fpt} we show that, when considering a particular choice of a bramble $\mathcal{B}$ which is constructed from $(k, r)$-linked sets, for appropriate choices of $k$ and $r$, we can decide if a given set $X$ is a hitting set of $\mathcal{B}$ in polynomial time and compute hitting sets of $\mathcal{B}$ in {\sf FPT} time when parameterized by $\ord(\mathcal{B})$.
Then, we show how to obtain a path $P$ intersecting all elements of $\mathcal{B}$ in polynomial time, improving Proposition~\ref{lemma:directed_grid_path_hitting_set}.
We use this latter result to give an {\sf FPT} algorithm with parameter $\ord(\mathcal{B})$ that produces, from a path $P$ intersecting all elements of a bramble of large order, a well-linked set $A$ of size $k$ which is contained in $V(P)$.
\begin{restatable}[Second main contribution]{theorem}{secondmain}
\label{theorem:path_sec_contribution}
Let $g(k) = (k+1)(\lfloor k/2 \rfloor +1) - 1$, $D$ be a digraph and $T$ be a $(g(k) - 1, g(k) - 1)$-linked set in $D$ with $|T| = 2g(k)-1$.
There is an algorithm running in time $2^{\Ocal(k^2 \log k)}\cdot n^{\Ocal(1)}$ that finds in $D$ a bramble $\mathcal{B}$ of order $g(k)$, a path $P$ that is a hitting set of $\mathcal{B}$, and a well-linked set $A$ of order $k$ such that $A \subseteq V(P)$.
\end{restatable}
The request that we make on $\ord(\mathcal{B})$ is also an improvement when compared to Proposition~\ref{lemma:directed_grid_bramble_path}.
In the next section we give an overview of how a cylindrical grid is found in~\cite{Kawarabayashi:2015:DGT:2746539.2746586} from the output of Theorem~\ref{theorem:path_sec_contribution}.
We discuss why the algorithms used in the remaining constructive steps of their proof are naturally \textsf{FPT} to obtain the following corollary, which is an improvement of Theorem~\ref{theorem:directed_grid_theorem}.

\begin{corollary}\label{theorem:second_contribution}
Let $k$ be a non-negative integer and $D$ be a digraph. There is a function $f: \mathbb{N} \to \mathbb{N}$ and an {\sf FPT} algorithm, with parameter $k$, that either
\begin{enumerate}
  \item produces an arboreal decomposition of $D$ of width at most $f(k)$, or
  \item finds a cylindrical grid of order $k$ as a butterfly minor of $D$.
\end{enumerate}
\end{corollary}


\subsection{Finding a cylindrical grid}\label{section:finding_a_cylindrical_grid}
On a very high level, the proof of the Directed Grid Theorem~\cite{Kawarabayashi:2015:DGT:2746539.2746586} can be summarized into the following three steps.
Using the terminology adopted in this paper, for a function $f$ as in the statement of Theorem~\ref{theorem:directed_grid_theorem} and given a digraph $D$, we
\smallskip
\begin{itemize}
  \item [(1)] pipeline Theorem~\ref{theorem:first_contribution_extended} and Theorem~\ref{theorem:path_sec_contribution} to either produce an arboreal decomposition of $D$ of width at most $f(k)$ or construct $\mathcal{B}$, $P$, and $A$ as in the statement of the latter;
  \item[(2)] use $P$ and $A$ to construct a well-linked \emph{path system} that is formed by a collection of paths; and 
  \item[(3)] iteratively refine the paths in the path system into new structures until a (butterfly) model of a cylindrical grid is obtained.
\end{itemize}
\smallskip
\begin{figure}[h!]
\centering
\scalebox{.8}{\begin{tikzpicture}[shorten >=1pt,->]]
\node[rectangle,draw, scale=1,align=center] (P-a) at (3*0,0) {$D,f(k)$};
\node[rectangle,draw,scale=1,align=center] (P-b) at (4*1.5,0) {$(f(k)-1, f(k)-1)$-\\linked set $T$};
\node[rectangle,draw,scale=1,align=center] (P-c) at (4*1.5,-3.5) {Path $P$,\\well-linked set $A \subseteq V(P)$};
\node[rectangle,draw,scale=1,text width=2.5cm,align=center] (P-d) at (4*3,-3.5) {Path system};
\node[scale=1.3,align=center] (P-e) at (4*3,0) {\scalebox{.25}{\usebox{\cydgrid}}};
\node[rectangle,draw,scale=1,align=center] (P-f) at (3*0,-3.5) {Arboreal decomposition\\of width $\leq f(k)$};

\draw[thick] (P-a) -- (P-b) node[pos=.5] (P-p2) {};
\draw[thick] (P-a) -- (P-f) node[pos=.5] (P-p1) {};
\draw[thick] (P-a) -- (P-f) node[pos=.5, xshift = 2.5cm, align=center] (P-p3) {Theorem \ref{theorem:first_contribution_extended}\\(1)};

\draw[dotted] (P-p3) -- (P-p1);
\draw[dotted] (P-p3) -- (P-p2);

\node[rectangle, draw, scale=1, align=center] (P-r) at ($(P-d) + (0,1.25)$) {Refinement};
\node (P-s3) at ($(P-r) + (1.5,0)$) {(3)};
\draw[thick, dashed, -] (P-d) -- (P-r);
\draw[thick, dashed] (P-r) -- (P-e);
\phantom{\draw[thick] (P-b) -- (P-c) node[pos=.5] (P-x) {};}
\phantom{\draw[thick] (P-c) -- (P-d) node[pos=.5] (P-y) {};}
\draw[thick] (P-b) -- (P-c) node[pos=.5, xshift = 2cm, align=center] (P-x1) {Theorem~\ref{theorem:path_sec_contribution}\\(2)};
\draw[thick] (P-c) -- (P-d); 
\draw[dotted] (P-x1) -- (P-x);

\end{tikzpicture}}%
\caption{Illustration of steps (1)-(2)-(3).}
\label{fig:automata_steps_finding_grid}
\end{figure}
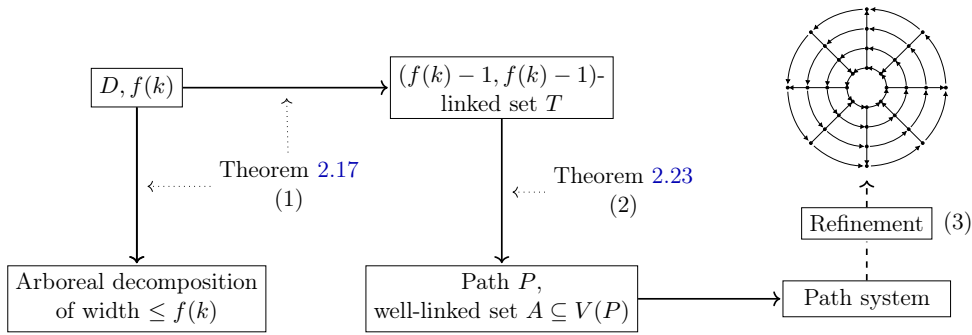
See Figure~\ref{fig:automata_steps_finding_grid} for an illustration of those steps.
As previously mentioned, we only improve on the procedures related to step~(1) and, in this section, we justify why this is sufficient to obtain Corollary~\ref{theorem:second_contribution}.
The main observations are that the algorithm runs maintaining and refining a collection of paths, where the size of the collection depends only on $k$, and that each of those refinements can be realized by iteratively testing how a given path intersects some subset of the collection.
The number of tests depends only on $k$ and each test is done in polynomial time.
We discuss here how to construct a path system from $P$ and $A$, as mentioned in step~(2) above.
For our examples, it is convenient to adopt the following definitions  from the full version of~\cite{Kawarabayashi:2015:DGT:2746539.2746586}.
\begin{definition}[Linkages]
Let $D$ be a digraph and $A,B \subseteq V(D)$ with $A \neq B$.
A \emph{linkage from $A$ to $B$} in $D$, or an \emph{$(A, B)$-linkage}, is a set of of pairwise vertex-disjoint paths from $A$ to $B$.
\end{definition}
\begin{definition}[Path system]
Let $D$ be a digraph and $\ell, p$ be two positive integers.
An \emph{$\ell$-linked path system of order $p$} is a sequence $\mathcal{S}$ with $\mathcal{S} = (\mathcal{P}, \mathcal{L}, \mathcal{A})$ where

\begin{itemize}
\item $\mathcal{P}$ is a sequence $P_1, \ldots, P_p$ of pairwise vertex-disjoint paths such that, for all $i \in [p]$, $V(P_i) \supseteq A_i^{\textsf{\emph{in}}} \cup A_i^{\textsf{\emph{out}}}$ and every vertex in $A_i^{\textsf{\emph{in}}}$ appears in $P_i$ before any vertex of $A_i^{\textsf{\emph{out}}}$; 

\item $\mathcal{L}$ is a collection $\{L_{i,j} \mid i,j \in [p] \text{ with } i \neq j\}$ of linkages where each $L_{i,j}$ is a linkage of size $\ell$ from $A_i^{\textsf{\emph{out}}}$ to $A_j^{\textsf{\emph{in}}}$; and

\item $\mathcal{A} = \{A_i^{\textsf{\emph{in}}}, A_i^{\textsf{\emph{out}}} \mid i \in [p]\}$ where each $A_i^{\textsf{\emph{in}}}$ and each $A_i^{\textsf{\emph{out}}}$ is a well-linked set of order $\ell$;
\end{itemize}
\end{definition}
Although the definition of path systems is quite loaded, it is not hard to visualize; see Figure~\ref{fig:path_system} for an illustration. 
\begin{figure}[h!]
\centering
\scalebox{.6}{
\begin{tikzpicture}[scale=.8]

\node[blackvertex, scale=.5] (u1) at (0,0) {};
\node[blackvertex, scale=.5] (u0) at ($(u1) + (0,6)$) {};

\node[blackvertex, scale=.5] (v0) at (12,0) {};
\node[blackvertex, scale=.5] (v1) at ($(v0) + (0,6)$) {};

\node[blackvertex, scale=.5] (w0) at (3,-2) {};
\node[blackvertex, scale=.5] (w1) at ($(w0) + (6,0)$) {};

\draw[-{Latex[length=3mm, width=2mm]}, shorten >= .1cm] (u0) -- (u1) node [midway] (um) {} node [pos = .25] (u0um) {} node [pos = .75] (umu1) {};
\draw ($(um) + (-0.5, 0)$) -- ($(um) + (0.5, 0)$);
\node (Ain) at ($(u0um) + (-1, 0)$) {\Large$A_1^{\textsf{in}}$};
\node (Ain) at ($(umu1) + (-1, 0)$) {\Large$A_1^{\textsf{out}}$};

\draw[-{Latex[length=3mm, width=2mm]}, shorten >= .1cm] (v0) -- (v1) node [midway] (vm) {} node [pos = .25] (v0vm) {} node [pos = .75] (vmv1) {};
\draw ($(vm) + (-0.5, 0)$) -- ($(vm) + (0.5, 0)$);
\node (Ain) at ($(v0vm) + (1, 0)$) {\Large$A_3^{\textsf{in}}$};
\node (Ain) at ($(vmv1) + (1, 0)$) {\Large$A_3^{\textsf{out}}$};

\draw[-{Latex[length=3mm, width=2mm]}, shorten >= .1cm] (w0) -- (w1) node [midway] (wm) {} node [pos = .25] (w0wm) {} node [pos = .75] (wmw1) {};
\draw ($(wm) + (0, -0.5)$) -- ($(wm) + (0, 0.5)$);
\node (Ain) at ($(w0wm) + (0, -1)$) {\Large$A_2^{\textsf{in}}$};
\node (Ain) at ($(wmw1) + (0, -1)$) {\Large$A_2^{\textsf{out}}$};

\draw[-{Latex[length=3mm, width=6mm]}, line width=4, shorten >= .1cm] ($(umu1) + (0,.5)$) -- ($(v0vm) + (0, .5)$);
\draw[-{Latex[length=3mm, width=6mm]}, line width=4, shorten >= .1cm] (umu1) -- (w0wm);

\draw[-{Latex[length=3mm, width=6mm]}, line width=4, shorten >= .1cm] ($(vmv1) + (0, .5)$) -- ($(u0um) + (0,.5)$);
\draw[-{Latex[length=3mm, width=6mm]}, line width=4, shorten >= .1cm] (vmv1) -- (w0wm);

\draw[-{Latex[length=3mm, width=6mm]}, line width=4, shorten >= .1cm] (wmw1) -- (u0um);
\draw[-{Latex[length=3mm, width=6mm]}, line width=4, shorten >= .1cm] (wmw1) -- (v0vm);

\end{tikzpicture}%
}
\caption{An $\ell$-linked path system of order $p =3$. A thick edge denotes a linkage of size $\ell$ from a set $A_i^{\text{out}}$ to a set $A_j^{\text{in}}$, with $i \neq j$.}
\label{fig:path_system}
\end{figure}
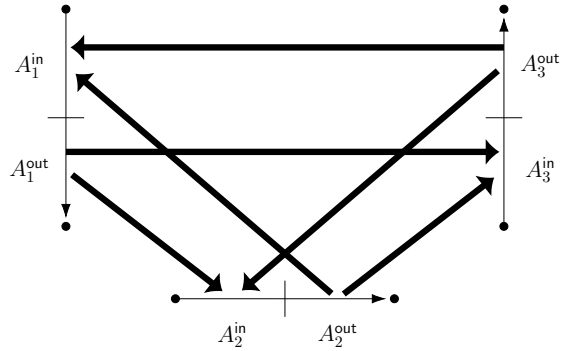
Notice that, knowing that the sets $A^i_{\textsf{in}}, A^i_{\textsf{out}}$ are well-linked, a path system is entirely formed by paths behaving in a particular way: the collection $\mathcal{P}$ of size $p$, and the collection of paths appearing in the linkages $L_{i,j}$.
Since each of those linkages has size $\ell$, an $\ell$-linked path system of order $p$ is formed by $p + 2\binom{p}{2}\ell$ paths.
With this observation, the task of constructing a path system from the output of Theorem~\ref{theorem:path_sec_contribution} becomes an easy one, as we proceed to explain.

Assume that we are given a path $P$ and a well-linked set $A$ with $|A| = 2\ell \cdot p$ and $A \subseteq V(P)$.
Let $\sigma = a_1, a_2, \ldots, a_{2\ell\cdot p}$ be an ordering of the vertices of $A$ as they appear in $P$, from the first to the last vertex of the path.
To construct an $\ell$-linked path system of order $p$, we follow $P$ in this order and, for $i \in [p]$, we define the path $P_i$ to be the subpath of $P$ from $a_{(i-1)2\ell+1}$ to $a_{i\cdot2\ell}$.
See Figure~\ref{fig:constructing_path_system} for an illustration of this procedure.
Since we know that $A$ is well-linked, and clearly every subset of a well-linked set is also well-linked, we define $A_i^{\textsf{in}}$ to be the set containing the first $\ell$ vertices of $V(P_i) \cap A$ and $A_i^{\textsf{out}}$ to be last $\ell$ vertices of $V(P_i) \cap A$ with respect to $\sigma$.

\begin{figure}[h!]
\centering
\scalebox{.45}{
\begin{tikzpicture}

\node[blackvertex, scale=.8, label={\huge $a_1$}] (a1) at (0,0) {};
\node[blackvertex, scale=.8, label={\huge $a_{2\ell+1}$}] (a2) at (6,0) {};
\node[blackvertex, scale=.8, label={\huge $a_{4\ell+1}$}] (a3) at (12,0) {};
\node[blackvertex, scale=.8, label={\huge $a_{6\ell+1}$}] (a4) at (18,0) {};
\node (ax) at ($(a4) + (1.5, 0)$) {\Huge$\cdots$};

\foreach \i in {1,...,3}{
  \node[blackvertex, scale=.5] (P-\i) at ($(a1) + (\i *1.5, 0)$) {};
}
\foreach \i in {1,...,3}{
  \node[blackvertex, scale=.5] (P-\i) at ($(a2) + (\i *1.5, 0)$) {};
}
\foreach \i in {1,...,3}{
  \node[blackvertex, scale=.5] (P-\i) at ($(a3) + (\i *1.5, 0)$) {};
}

\draw[-{Latex[length=3mm, width=2mm]}, shorten >= .1cm] (a1) -- (a2);
\draw[-{Latex[length=3mm, width=2mm]}, shorten >= .1cm] (a2) -- (a3);
\draw[-{Latex[length=3mm, width=2mm]}, shorten >= .1cm] (a3) -- (a4);
\draw[-{Latex[length=3mm, width=2mm]}, shorten >= .1cm] (a4) -- (ax);

\draw [decorate,decoration={brace,amplitude=10pt, mirror},xshift=0pt,yshift=0]
($(a1) + (0,-.75)$) -- ($(a2) + (-1.5,-.75)$) node [black,midway,yshift=-.85cm] {\Huge $P_1$};

\draw [decorate,decoration={brace,amplitude=10pt, mirror},xshift=0pt,yshift=0]
($(a2) + (0,-.75)$) -- ($(a3) + (-1.5,-.75)$) node [black,midway,yshift=-.85cm] {\Huge $P_2$};

\draw [decorate,decoration={brace,amplitude=10pt, mirror},xshift=0pt,yshift=0]
($(a3) + (0,-.75)$) -- ($(a4) + (-1.5,-.75)$) node [black,midway,yshift=-.85cm] {\Huge $P_3$};
\end{tikzpicture}%
}%
\caption{Finding the paths $P_i$ from $P$ and $A$, for $i \in [p]$.}
\label{fig:constructing_path_system}
\end{figure}
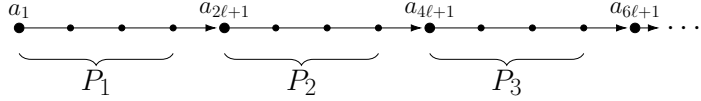

Next, for $i,j \in [p]$ with $i \neq j$, we choose $L_{i,j}$ to be a linkage from $A_i^{\textsf{out}}$ to $A_j^{\textsf{in}}$.
At least one choice for $L_{i,j}$ is guaranteed to exist because $A_i^{\textsf{in}} \cup A_j^{\textsf{out}} \subseteq A$ and $A$ is well-linked.
Moreover, we can find each linkage in polynomial time by applying Menger's Theorem (cf. Theorem~\ref{thm:Menger}) and solving a flow problem.
Hence, given $P$ and $A$ of adequate size, we can find an $\ell$-linked path system of order $p$ in polynomial time.

As in Figure~\ref{fig:grid_in_system}, it is easy to find a cylindrical grid in a sufficiently large path system if it is ``well-behaved'', that is, when the paths in $\mathcal{L}$ are pairwise internally vertex-disjoint. 
In fact, in such cases every $P_i$ models one vertex of a \emph{biclique} that is a butterfly minor of $D$.
A biclique is a digraph $H$ having a pair of edges in both directions between any two vertices of $H$.
Clearly, a biclique with $2k^2$ vertices contains a cylindrical grid of order $k$.

\begin{figure}[h!]
\centering
\begin{subfigure}{.6\textwidth}
\centering
\scalebox{.4}{\begin{tikzpicture}[scale=.8, yscale=.8]

\node[blackvertex, scale=.75] (u0) at (0,0) {};
\node[blackvertex, scale=.75] (u1) at ($(u0) + (7,0)$) {};
\node[blackvertex, scale=.75] (u2) at ($(u1) + (2,0)$) {};
\node[blackvertex, scale=.75] (u3) at ($(u2) + (7,0)$) {};

\node[blackvertex, scale=.75] (v0) at (0,8) {};
\node[blackvertex, scale=.75] (v1) at ($(v0) + (7,0)$) {};
\node[blackvertex, scale=.75] (v2) at ($(v1) + (2,0)$) {};
\node[blackvertex, scale=.75] (v3) at ($(v2) + (7,0)$) {};

\draw[-{Latex[length=3mm, width=2mm]}, shorten >= .1cm, densely dashed, thick] (u0) -- (u1) node [midway] (u0u1) {} node [pos = .25,label = 270:{\huge$1'$},blackvertex,scale=.35] (uin1) {} node [pos = .75,label = 270:{\huge$2'$},blackvertex,scale=.35] (uout1) {};
\node (in) at ($(uin1) + (-1, -2)$) {\huge in};
\node (out) at ($(uout1) + (1, -2)$) {\huge out};

\draw[-{Latex[length=3mm, width=2mm]}, shorten >= .1cm, densely dashed, thick] (u2) -- (u3) node [midway] (u2u3) {} node [pos = .25,label = 270:{\huge$3'$},blackvertex,scale=.35] (uin2) {} node [pos = .75,label = 270:{\huge$4'$},blackvertex,scale=.35] (uout2) {};
\node (in) at ($(uin2) + (-1, -2)$) {\huge in};
\node (out) at ($(uout2) + (1, -2)$) {\huge out};

\draw ($(u0u1) + (0, -0.5)$) -- ($(u0u1) + (0, 0.5)$);
\draw ($(u2u3) + (0, -0.5)$) -- ($(u2u3) + (0, 0.5)$);

\draw[-{Latex[length=3mm, width=2mm]}, shorten >= .1cm, densely dashed, thick] (v0) -- (v1) node [midway] (v0v1) {} node [pos = .25,label = 90:{\huge$2$},blackvertex,scale=.35] (vin1) {} node [pos = .75,label = 90:{\huge$3$},blackvertex,scale=.35] (vout1) {};
\node (in) at ($(vin1) + (-1, 2)$) {\huge in};
\node (out) at ($(vout1) + (1, 2)$) {\huge out};
\draw[-{Latex[length=3mm, width=2mm]}, shorten >= .1cm, densely dashed, thick] (v2) -- (v3) node [midway] (v2v3) {} node [pos = .25,label = 90:{\huge$4$},blackvertex,scale=.35] (vin2) {} node [pos = .75,label = 90:{\huge$1$},blackvertex,scale=.35] (vout2) {};
\node (in) at ($(vin2) + (-1, 2)$) {\huge in};
\node (out) at ($(vout2) + (1, 2)$) {\huge out};

\draw ($(v0v1) + (0, -0.5)$) -- ($(v0v1) + (0, 0.5)$);
\draw ($(v2v3) + (0, -0.5)$) -- ($(v2v3) + (0, 0.5)$);

\draw[-{Latex[length=3mm, width=2mm]}, shorten >= .1cm] (uout1) to [out=135, in = -90] (vin1);
\draw[-{Latex[length=3mm, width=2mm]}, shorten >= .1cm] (vout1) to [out = -45, in = 90] (uin2);
\draw[-{Latex[length=3mm, width=2mm]}, shorten >= .1cm] (uout2) to [out = 135, in = -90] (vin2);
\draw[-{Latex[length=3mm, width=2mm]}, shorten >= .1cm] (vout2) to [out = -90, in = 90] (uin1);

\path[densely dashed, thick]
  (v3) edge ($(v3) + (0,3)$)
  ($(v3) + (0,3)$) edge ($(v0) + (0,3)$)
  ($(v0) + (0,3)$) edge [-{Latex[length=3mm, width=2mm]}]  (v0);

\path[densely dashed, thick]
  (u3) edge ($(u3) + (0,-3)$)
  ($(u3) + (0,-3)$) edge ($(u0) + (0,-3)$)
  ($(u0) + (0,-3)$) edge [-{Latex[length=3mm, width=2mm]}]  (u0);

\draw[-{Latex[length=3mm, width=2mm]},densely dashed, thick] (v1) to [out = 45, in = 135] (v2);
\draw[-{Latex[length=3mm, width=2mm]},densely dashed, thick] (u1) to [out = 45, in = 135] (u2);
\end{tikzpicture}}%
\end{subfigure}\hspace{1cm}
\begin{subfigure}{.3\textwidth}
\scalebox{.45}{\begin{tikzpicture}

\def \n {4}
\def \margin {3}
\def \radiusinner {1.35cm}
\foreach \i/\l in {1/180, 2/270, 3/0, \n/90}
{
  \setcounter{c}{\i};
  \def \vname {\alph{c}};

  \node[blackvertex,vertex,scale=.4, label=\l:{\huge$\i$}] (P-\vname) at ({360/\n * (\i - 1)}:\radiusinner) {};
  \draw[-{Latex[length=2mm, width=2mm]}, densely dashed, thick] ({360/\n * (\i - 1)+\margin}:\radiusinner)
  arc ({360/\n * (\i - 1)+\margin}:{360/\n * (\i)-\margin}:\radiusinner);
}

\def \radiusinner {3.3cm}

\foreach \i/\l in {1/0, 2/90, 3/180, \n/270}
{
  \setcounter{c}{\i};
  \def \vname {\Alph{c}};

  \node[blackvertex,vertex,scale=.4, label=\l:{\huge$\i'$}] (P-\vname) at ({360/\n * (\i - 1)}:\radiusinner) {};
  \draw[-{Latex[length=2mm, width=2mm]}, densely dashed, thick] ({360/\n * (\i - 1)+\margin}:\radiusinner)
  arc ({360/\n * (\i - 1)+\margin}:{360/\n * (\i)-\margin}:\radiusinner);
}

\draw[-{Latex[length=2mm, width=2mm]}] (P-a) to (P-A);
\draw[-{Latex[length=2mm, width=2mm]}] (P-B) to (P-b);
\draw[-{Latex[length=2mm, width=2mm]}] (P-c) to (P-C);
\draw[-{Latex[length=2mm, width=2mm]}] (P-D) to (P-d);

\end{tikzpicture}}%
\end{subfigure}
\caption{An example of cylindrical grid of order two in a ``well-behaved'' path system, where we assume that the paths in $\mathcal{L}$ are pairwise internally vertex-disjoint.} 
\label{fig:grid_in_system}
\end{figure}
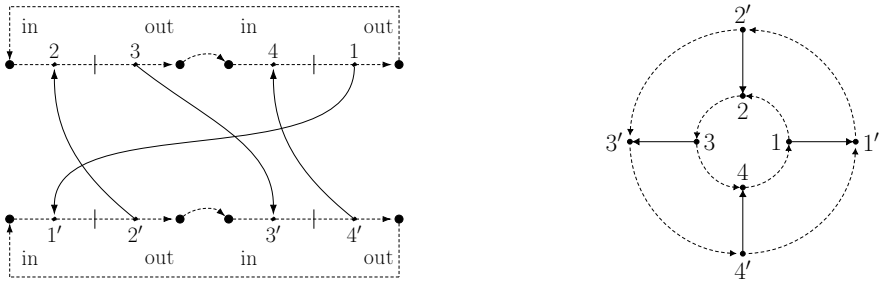

Unfortunately, in general we cannot expect every path system to behave in this way.
Hence, the proof of the Directed Grid Theorem by Kawarabayashi and Kreutzer~\cite{Kawarabayashi:2015:DGT:2746539.2746586} follows a sequence of refinements, as mentioned in item $(3)$ above, each constructing a new structure from the previous one until a cylindrical grid is obtained.
This part is represented by the dashed edges in Figure~\ref{fig:automata} and Figure~\ref{fig:automata_steps_finding_grid} and, although it is not hard to see that the algorithms realizing those constructions are naturally \textsf{FPT}, the constructive proofs are in fact the largest and most involved part of their paper.
Namely, they show how to find a \emph{web}\footnote{\label{footnote:webs_and_fences}The definitions of \emph{webs} and \emph{fences} can be found in the full version of~\cite{Kawarabayashi:2015:DGT:2746539.2746586}.} or a cylindrical grid from a path system that is sufficiently large.
If a web is obtained, then the next step is to find a \emph{fence}\cref{footnote:webs_and_fences} in it.
Lastly, they prove that we are guaranteed to find a cylindrical grid of order $k$ in any sufficiently large fence.

Fortunately, and as it is the case with path systems, webs and fences are defined around collections of paths satisfying some properties that can be easily verified in polynomial time.
Since the number of paths in a $\ell$-linked path system of order $p$ depends only on $\ell$ and $p$, we can search for a web in a path system by testing the defining properties of webs for every subset of the set of paths in the path system.
Thus, in \textsf{FPT} time with parameters $\ell$ and $p$ we can find a web in a path system.
A similar approach is viable to find fences in webs and cylindrical grids in fences and thus Corollary~\ref{theorem:second_contribution} follows from Theorem~\ref{theorem:path_sec_contribution}.


\section{Balanced separators and arboreal decompositions}
\label{sec:FPT-arboreal}

The algorithm for arboreal decompositions given in~\cite{Johnson.Robertson.Seymour.Thomas.01} starts with a trivial decomposition $(\{r\}, \emptyset, \{W_r\})$ whose underlying arborescence contains only one vertex $r$. Thus, $W_r = V(G)$.
Each iteration splits the vertices contained in an excessively large leaf of the current decomposition, if one exists, into a set of new leaves, while guaranteeing that the width of the non-leaf vertices remains bounded from above by a function of $k$.
Although this problem is not explicitly named by the authors, on each of those split operations the algorithm has to decide whether the input digraph admits a $(T, r)$-balanced separator for a given set $T$.
Formally, on each iteration the need to solve a particular case of the following problem.

\begin{boxproblem}[framed]{\textsc{Balanced Separator}}\label{problem:partitioning_set}
Input: & A digraph $D$, a set $T \subseteq V(D)$ of size $k$, and two non-negative integers $r$ and $s$.\\
Output: & A $(T,r)$-\balsep $Z$ with $|Z| \leq s$, if it exists.
\end{boxproblem}

The \textsc{Balanced Separator} problem can be naively solved by checking all $\binom{n}{s}$ sets $Z$ of size $s$ in $V(D)$ and enumerating the strong components of $D \setminus Z$.
Therefore it is in {\sf XP} with parameter $s$.
Furthermore, the process of finding balanced separators is the only step of the algorithm given in~\cite{Johnson.Robertson.Seymour.Thomas.01} that is done in \textsf{XP} time.
In the next section, we show how to compute $(T,r)$-\balsep{s} in \textsf{FPT} time with parameter $k$.
In particular, we show that a set $Z$ is a $(T,r)$-\balsep if and only if $Z$ is a solution to a separation problem introduced in~\cite{Erbacher.Jaeger.14} that is a particular case of the \textsc{Multicut} problem in digraphs.
Then, we use this result to improve the algorithm by Johnson et al.~\cite{Johnson.Robertson.Seymour.Thomas.01} for approximate arboreal decompositions (cf. Proposition~\ref{theorem:directed_tree_width}), showing that it can be done in {\sf FPT} time.
Notice that we can assume that $r \leq k - 1$ and $s \leq k - r - 1$: if $r \geq k$, the empty set is a $(T,r)$-\balsep and, if $s \geq k - r$, any choice of $s$ vertices from $T$ form a $(T,r)$-\balsep.
To avoid repetition, we make these considerations here and refrain from repeating them in the remainder of this article.
We refer to instances of \textsc{Balanced Separator} as $(D, T, k, r, s)$.

\subsection{Computing \texorpdfstring{$(T,r)$-\balsep{s}}{(T,r)-Partitioning sets} in {\sf FPT} time}
Given a graph or digraph  $D$ and a set of pairs of terminal vertices $\{(s_1, t_1)$, $(s_2, t_2)$, $\ldots, (s_k, t_k)\}$, the \textsc{Multicut} problem asks to minimize the size of a set $Z \subseteq V(D)$ such that there is no path from $s_i$ to $t_i$ in $D \setminus Z$, for $i \in [k]$.
When parameterized by the size of the solution, the problem is {\sf FPT} in undirected graphs~\cite{doi:10.1137/110855247,doi:10.1137/140961808}.
On the directed case, this problem is \textsf{FPT} in DAGs when parameterized by the size of the solution and the number of pairs of terminals~\cite{doi:10.1137/120904202}, but \textsc{W}[1]-hard in the general case even for fixed $k=4$~\cite{Pilipczuk:2018:DMW:3208319.3201775}.

A variation of \textsc{Multicut} is considered in~\cite{Erbacher.Jaeger.14}.
Namely, in the \textsc{Linear Edge Cut} problem, we are given a digraph $D$ and a collection of sets of vertices $\{S_1, \ldots, S_k\}$, and we want to find a minimum set of edges $Z$ such that there is no path from $S_i$ to $S_j$ in $D \setminus Z$ whenever $j > i$.
We remark that the authors in~\cite{Erbacher.Jaeger.14} refer to this problem as \textsc{Linear Cut} only.
This problem is {\sf FPT} when parameterized by the size of the solution:

\begin{proposition}[Erbacher et al.~\cite{Erbacher.Jaeger.14}]
\label{theorem:fpt_edge_linear_cut}
The \textsc{Linear Edge Cut} problem can be solved in time $\Ocal(4^s \cdot s \cdot n^4)$, where $s$ is the size of the solution.
\end{proposition}
We remark that the authors of~\cite{Erbacher.Jaeger.14} mention that this result can also be achieved by using a reduction to the \textsc{Skew Separator} algorithm given in~\cite{Chen.Liu.Lu.Sullivan.Razgon.2008}.


In this section, we show how to use the algorithm for the \textsc{Linear Edge Cut} problem to solve the vertex version, and then show how this version can be used to compute $(T,r)$-\balsep{s} in {\sf FPT} time.
We formally define the vertex version below.


\begin{boxproblem}[framed]{\textsc{Linear Vertex Cut}}\label{problem:linear_vertex_cut}
Input: & A digraph $D$, a collection of terminal sets $\mathcal{T}$, with $\mathcal{T} = \{T_1, T_2, \ldots, T_k\}$, where $T_i \subseteq V(D)$ for $i \in [k]$, and an integer $s \geq 0$.\\
Question: & Is there a set of vertices $Z \subseteq V(D)$ with $|Z| \leq s$ such that there are no paths in $D \setminus Z$ from $T_i$ to $T_j$, for $1 \leq i < j \leq k$?
\end{boxproblem}
From an instance $(D, \mathcal{T}, s)$ of \textsc{Linear Vertex Cut}, we construct an equivalent instance of $(D', \mathcal{T'}, s)$ of \textsc{Linear Edge Cut} as follows.
First, notice that any vertex $v$ occurring in the intersection of two distinct sets in $\mathcal{T}$ must be part of any solution for the instance.
Thus we can assume that every vertex of $D$ occurs in at most one set in $\mathcal{T}$.
Now, for each vertex $v \in V(D)$, add to $D'$ two vertices $v_{\sf in}$ and $v_{\sf out}$ and an edge $e_v$ from $v_{\sf in}$ to $v_{\sf out}$.
For each edge $e \in E(D)$ with tail $u$ and head $v$, add to $D'$ a set of $s+1$ parallel edges from $u_{\sf out}$ to $v_{\sf in}$.
Finally, for each $v \in T_i$, for $i \in [k]$, add a new vertex $v'$ to $D'$ together with $s+1$ edges from $v'$ to $v_{\sf in}$ and $s+1$ edges from $v_{\sf out}$ to $v'$.
Let $T'_i = \{v' \mid v \in T_i\}$ and $\mathcal{T}' = \{T'_1, \ldots, T'_k\}$.
We have the following easy lemma.
\begin{lemma}\label{lemma:equivalence_linear_edge_vertex}
An instance $(D, \mathcal{T}, s)$ of \textsc{Linear Vertex Cut} is positive if and only if the associated instance $(D', \mathcal{T}', s)$ of \textsc{Linear Edge Cut} is positive.
\end{lemma}
\begin{proof}
Let $Z \subseteq V(D)$ be a solution for $(D, \mathcal{T},s)$ and $Z' = \{e_v \mid v \in Z\} \subseteq E(D')$.
By contradiction, assume that there is a path $P'$ in $D' \setminus Z'$ from a vertex $u'$ to a vertex $v'$, for $u' \in T'_i$, $v' \in T'_j$, and $j > i$. Then there is a path $P$ from $u$ to $v$ in $D \setminus Z$ with vertex set $\{v \mid e_v \in E(P')\}$. This contradicts our choice of $Z$ and thus the necessity holds.

For the sufficiency, let $Z'$ be a minimal solution for $(D', \mathcal{T}', s)$.
Notice that all edges in $Z'$ are from a vertex $v_{\sf in}$ to its respective $v_{\sf out}$, as the budget $s$ for the size of $Z'$ does not allow any other choice.
Let $Z = \{v \mid e_v \in Z'\}$ and, by contradiction, let $P$ be a path in $D \setminus Z$ from a vertex $u$ to a vertex $v$, with $u \in T_i$, $v \in T_j$, and $j > i$.
For each edge $e \in E(P)$ with $e = (x,y)$ there is an edge $e'$ with $e' = (x_{\sf out}, y_{\sf in})$ in $D' \setminus Z'$.
Let $F'$ be the set of such edges of $D'$.
Now, there is a path $P'$ from $u_{\sf in}$ to $v_{\sf out}$ in $D'$ with edge set $\{e_v \mid v \in V(P)\} \cup F'$.
Appending to $P'$ the edges from $u'$ to $u_{\sf in}$ and from $v_{\sf out}$ to $v'$ we construct a path from $u'$ to $v'$ in $D \setminus Z'$, contradicting our choice of $Z'$. Therefore, the sufficiency also holds and the lemma follows.
\end{proof}
\noindent Combining Proposition~\ref{theorem:fpt_edge_linear_cut} and Lemma~\ref{lemma:equivalence_linear_edge_vertex} we get the following.
\begin{corollary}\label{corollary:fpt_linear_vertex_cut}
There is an  {\sf FPT}  algorithm for the \textsc{Linear Vertex Cut} problem parameterized by the size $s$ of the solution and running in time $\Ocal(4^s \cdot s \cdot n^4)$.
\end{corollary}

We now show how to solve \textsc{Balanced Separator} using \textsc{Linear Vertex Cut}.
Namely, we show that a digraph $D$ admits a $(T,r)$-\balsep $Z$ if and only if $Z$ is a solution to some instance $(D, \mathcal{T},s)$ of \textsc{Linear Vertex Cut} where $\mathcal{T}$ depends of $T$.



\begin{lemma}\label{lemma:problem_equivalence}
Let $(D, T, k, r, s)$ be an instance of \textsc{Balanced Separator}.
A set $Z \subseteq V(D)$ with $|Z| \leq s$ is a $(T,r)$-\balsep if and only if there is a partition $\mathcal{T}$ of $T$ into sets $T_1, T_2, \ldots, T_\ell$ such that
\begin{enumerate}
  \item $|T_i| \leq r$, for $i \in [\ell]$, and
  \item $Z$ is a solution for the instance $(D, \mathcal{T}, s)$ of  \textsc{Linear Vertex Cut}.
\end{enumerate}

\end{lemma}
\begin{proof}
For the necessity, let $Z$ be a $(T,r)$-\balsep with $|Z| \leq s$.
Let $\mathcal{C}$ be the set of strong components of $D \setminus Z$ and consider an ordering $C_1, \ldots, C_\ell$ of its elements such that there is no path from $C_i$ to $C_j$ in $D \setminus Z$ whenever $j > i$.
Notice that this is the reverse of a topological ordering for the elements of $\mathcal{C}$.
Let $v_1, \ldots, v_q$ be the vertices in $T \cap Z$, if any exist.
For $i \in [\ell]$, choose $T_i = V(C_i) \cap T$ and define $\mathcal{T} = \{T_1, T_2, \ldots, T_\ell\}$ if $T \cap Z \neq \emptyset$ or $\mathcal{T} = \{T_1, T_2, \ldots, T_\ell, \{v_1\}, \ldots, \{v_q\}\}$ otherwise.
Notice that it is possible for a set $T_i$ to be empty.

Since $Z$ is a $(T,r)$-\balsep, we know that $|T_i| \leq r$ holds for all $i \in [\ell]$.
Since the vertices in a non-empty set $T_i$ are contained in exactly one strong component of $D \setminus Z$, any path between different sets in $\mathcal{T}$ must contain a path between distinct strong components of $D \setminus Z$.
Thus we conclude that there are no paths from a set $T_i$ to another set $T_j$ with $j > i$, since otherwise we would have a contradiction to our choice for the order of the elements of $\mathcal{C}$, and therefore $Z$ is a solution for the instance $(D, \mathcal{T}, s)$ of \textsc{Linear Vertex Cut}.

For the sufficiency, let $\mathcal{T}$ be as in the statement of the lemma and $Z$ be a solution for the instance $(D, \mathcal{T}, s)$ of \textsc{Linear Vertex Cut}.
First, notice that no strong component of $D \setminus Z$ can intersect two distinct sets $T,T' \in \mathcal{T}$.
Indeed, if this were the case, then there would be a path in $D \setminus Z$ from a vertex in $T$ to a vertex in $T'$ and vice-versa, contradicting the fact that $Z$ is a solution for $(D, \mathcal{T},s)$.
Thus, if $|V(C) \cap T| \geq r+1$ for some strong component $C$ of $D \setminus Z$, we have a contradiction as $C$ would intersect at least two distinct sets in $\mathcal{T}$.
We conclude that $Z$ is a $(T,r)$-\balsep and the lemma follows.
\end{proof}


The {\sf FPT} algorithm for \textsc{Balanced Separator} follows from Lemma~\ref{lemma:problem_equivalence} and Corollary~\ref{corollary:fpt_linear_vertex_cut}.
The running time is heavily tied to the number of partitions $\mathcal{T}$ that can be generated from a given set $T$ of an instance $(D, T, k, r, s)$ of \textsc{Balanced Separator}.
This value is bounded by the $k$-th \emph{ordered Bell number}~\cite{cayley_2009}. The \emph{Bell number}~\cite{10.2307/1968431} counts the number of partitions of a set, and its ordered variant also considers the number of possible orderings for each partition. The $k$-th ordered Bell number is of the form $2^{\Ocal(k \log k)}$.
From the previous discussion we get the following theorem.
\begin{theorem}\label{theorem:FPT_haven}
There is an algorithm running in time $\RunTimeAlg$ for the \textsc{Balanced Separator} problem.
\end{theorem}
\begin{proof}
Let $(D, T, k, r, s)$ be an instance of \textsc{Balanced Separator} and $\mathcal{T}^*$ be the set of all ordered partitions $\{T_1, \ldots, T_\ell\}$ of $T$ with $|T_i| \leq r$, for $i \in [\ell]$.

By Corollary~\ref{corollary:fpt_linear_vertex_cut}, we can solve instances of \textsc{Linear Vertex Cut} problem in time $\Ocal(4^s \cdot s \cdot n^4)$ for $s$ being the size of the solution.
By Lemma~\ref{lemma:problem_equivalence}, $Z$ is a $(T,r)$-\balsep if and only if there is a $\mathcal{T} \in \mathcal{T}^*$ such that the instance $(D, \mathcal{T},s)$ of \textsc{Linear Vertex Cut} is positive.
Finally, since $|\mathcal{T}^*|$ is at most the $k$-th ordered Bell number, we can solve \textsc{Balanced Separator} by testing $2^{\Ocal (k \log k)}$ instances of \textsc{Linear Vertex Cut}.
As $s \leq k-r$ (since otherwise the instance of \textsc{Balanced Separator} is trivially positive), the bound on the running time follows.
\end{proof}

\subsection{An \textsf{FPT} algorithm for approximate arboreal decompositions}\label{section:FPT_arb_decomp}

We are now ready to prove Theorem~\ref{theorem:first_contribution_intro}.
We remark that the proof below follows~\cite[3.3]{Johnson.Robertson.Seymour.Thomas.01} except that we replace the {\sf XP} procedure of the proof by our {\sf FPT} algorithm for \textsc{Balanced Separator}.
In the following proof, we need to test whether a given set $T \subseteq V(D)$ admits a $(T, k-1)$-\balsep of size at most $k-1$.
Thus we remind the reader of the discussion made in the beginning of Section~\ref{sec:FPT-arboreal}: if $|T| \leq 2k-2$, then the answer is positive since we can pick any $k-1$ vertices of $T$ to form a solution.

\firstmainklinked*

\begin{proof}
We begin with a nice arboreal decomposition $(R_0, \mathcal{X}_0, \mathcal{W}_0)$ of $D$ where $\mathcal{X}_0 = \emptyset$, $V(R_0) = \{r\}$, and $\mathcal{W}_0 = \{V(D)\}$.
We maintain an arboreal decomposition $(R,\mathcal{X,W})$ of $D$ for which the following two properties hold:
\begin{enumerate}
\item[(P1)] $|W_r \cup (\bigcup_{e \thicksim r} X_e)| \leq 3k-1$ for every $r \in V(R)$ of out-degree at least one, and
\item[(P2)] $|X_e| \leq 2k-1$ for every $e \in E(R)$.
\end{enumerate}
Notice that both (P1) and (P2) hold for $(R_0, \mathcal{X}_0, \mathcal{W}_0)$.

If (P1) holds for all $r \in V(R)$, then we have constructed an arboreal decomposition with the desired width. Otherwise, we can assume that $(R, \mathcal{X,W})$ contains at least one leaf that is \emph{too large}. That is, the width of a vertex $r_0$ of out-degree zero of $R$ is at least $3k$. If there is an edge $e_0 \in E(R)$ with head $r_0$, let $T = X_{e_0}$. Otherwise, let $T = \emptyset$. Either way, $|T| \leq 2k-1$ and $|W_{r_0}| \geq 3k - |T| \geq k+1$.

Now, we test whether $D$ contains a $(T,k-1)$-\balsep of size at most $k-1$ and, by Theorem~\ref{theorem:FPT_haven}, this test can be done in time $\RunTimeAlg$.
If $|T| \leq 2k-2$ then the answer is positive since we can pick any set of $k-1$ vertices of $T$ to form a solution.
Thus, if the answer is negative, we have $|T| = 2k-1$ and we terminate the algorithm outputting $T$.
We may now assume that $D$ contains a $(T,k-1)$-\balsep $Z'$ with $|Z'| \leq k-1$.

From the bound on the sizes of the sets, there are at least two vertices in $W_{r_0} \setminus Z'$.
Choose $v$ to be any of those two vertices, and let $Z = Z' \cup \{v\}$.
Now $|Z| \leq k$, $Z \cap W_{r_0} \neq \emptyset$, and $|V(C) \cap T| \leq k-1$ holds for every strong component $C$ of $D \setminus Z$.


Let $C_1, \ldots, C_\ell$ be the strong components of $D \setminus Z$. If $B$ is a strong component of $C_i \setminus T$, for $i \in [\ell]$, then either $V(B) \subseteq W_{r_0}$ or $V(B) \cap W_{r_0} = \emptyset$, for $W_{r_0}$ is $T$-guarded.
Let $B_1, \ldots, B_d$ be all such strong components for which $V(B_j) \subseteq W_{r_0}$ for all $j \in [d]$.
Furthermore, let $f: \mathbb{N} \to \mathbb{N}$ be a function assigning an index $j$ to an index $i$ if and only if $B_i \subseteq C_j \setminus T$. Thus, $f$ can be used to tell which set $C_j$ contains a given $B_i$. Now, $Z \cap W_{r_0}, V(B_1), \ldots, V(B_d)$ is a partition of $W_{r_0}$ into non-empty sets.
We show that this partition yields another arboreal decomposition of $D$.

Let $R'$ be the arborescence obtained from $R$ by adding a vertex $r_i$ and an edge $e_i$ from $r_0$ to $r_i$, for $i \in [d]$. Furthermore, let $X'_e = X_e$ for all $e \in E(R)$ and $W_r' = W_r$ for all $r \in V(R) \setminus \{r_0\}$. Also, let $W'_{r_0} = W_{r_0} \cap Z$ and, for $i \in [d]$, let $X'_{e_i} = Z \cup (V(C_{f(i)}) \cap T)$ and $W'_{r_i} = V(B_i)$.
Finally, define $\mathcal{X}' = \{X'_e \mid e \in E(R')\}$ and $\mathcal{W}' = \{W'_r \mid r \in V(R')\}$. As the vertices of $W_{r_0}$ have been spread into non-empty sets, we only need to verify that $(R', \mathcal{X', W'})$ is an arboreal decomposition of $D$ for which (P1) and (P2) hold; see Figure~\ref{fig:arboreal_decomposition_algorithm} for an illustration.

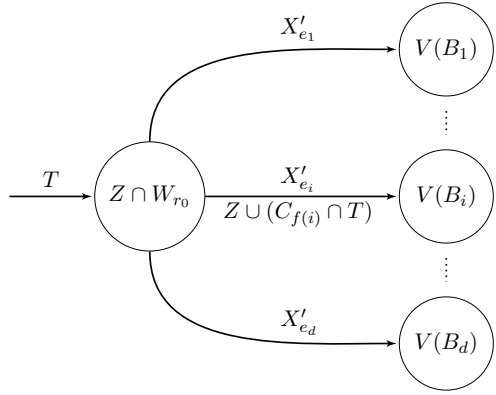
\begin{figure}[ht]
\centering
\scalebox{.78}{
\begin{tikzpicture}[shorten >=1pt]]

\node[state, scale=1,text width=1.5cm,align=center] (P-a) at (0,0) {$Z \cap W_{r_0}$};
\node[state,scale=1,text width=1.2cm,align=center] (P-b) at (5*1,-2.5) {$V(B_d)$};
\node[state,scale=1,text width=1.2cm,align=center] (P-c) at (5*1,0) {$V(B_i)$};
\node[state,scale=1,text width=1.2cm,align=center] (P-d) at (5*1,2.5) {$V(B_1)$};

\draw[edge, thick] (P-a) to [out = 90, in =180] node [pos=.7, above, sloped] {$X'_{e_1}$} (P-d);
\draw[edge, thick] (P-a) to [out = -90, in =180] node [pos=.7, above] {$X'_{e_d}$} (P-b);

\draw[edge, thick] (P-a) to (P-c);
\node (P-n) at (1*5/2,-0.3) {$Z \cup (C_{f(i)} \cap T)$};
\node (P-n) at (1*5/2, 0.3) {$X'_{e_i}$};
\node (P-t) at (-1*5/2,0) {};
\draw[thick,dotted, shorten <= .25cm, shorten >= .25cm] (P-d) -- (P-c);%
\draw[thick,dotted, shorten <= .25cm, shorten >= .25cm] (P-c) -- (P-b);%
\draw[edge,thick] (P-t) -- (P-a);
\node (P-tn) at (-1*5/3,0.3) {$T$};

\end{tikzpicture}%
}%
\caption{Spreading the vertices in $W_{r_0}$.}
\label{fig:arboreal_decomposition_algorithm}
\end{figure}

$\mathcal{W}'$ is indeed a partition of $V(D)$ into non-empty sets, as $W_{r_0}$ is partitioned into non-empty sets. For $i \in [d]$, $W'_{r_i} = V(B_i)$ and $B_i$ is a strong component of $C_{f(i)} \setminus T$.
Thus, each new leaf $r_i$ added to $R$ is such that $W'_{r_i}$ is $X'_{e_i}$-guarded and, for all $e \in E(R')$, $\bigcup\{W'_r : r \in V(R'), r > e\}$ is $X'_e$-guarded as the property remains unchanged for all $e \in E(R)$.

For $r \in V(R)$, the validity of (P1) remains unchanged.
The width of $r_0$ is bounded from above by $|T| + |Z| \leq 2k-1 + k  = 3k - 1$, as desired, for $W'_{r_0} \subseteq Z$ and $\bigcup_{e \thicksim r_0}X'_e \subseteq T \cup Z$.
(P2) remains true in $(R', \mathcal{X', W'})$ for all $e \in E(R)$. For $e_i$, $i \in [d]$, $|X_{e_i}| \leq |Z| + |V(C_{f(i)}) \cap T|$.
By the assumption that $ (D,T,2k-1,k-1,k-1)$ is a positive instance of \textsc{Balanced Separator}, $|Z| + |V(C_{f(i)} \cap T| \leq k + k - 1 = 2k-1$.

Observe that, since each $B_i$ is disjoint from $T \cup Z$, $(R', \mathcal{X', W'})$ is actually  a {\sl nice} arboreal decomposition.

Now, if no leaf of $(R', \mathcal{X', W'})$ is too large, we end the algorithm returning this arboreal decomposition of $D$.
Otherwise, we repeat the aforementioned procedure with new choices for $T$ and $W_{r_0}$.


Finally, the running time holds by Theorem~\ref{theorem:FPT_haven}, since $\mathcal{W}$ partitions $V(D)$ into non-empty sets and each iteration decreases the number of vertices in leaves that have width at least $3k$.
\end{proof}
The proof of Theorem~\ref{theorem:first_contribution_intro} easily follows from Lemma~\ref{claim:haven_instance} and Theorem~\ref{theorem:first_contribution_extended}.
\firstmain*
\begin{proof}
Applying Theorem~\ref{theorem:first_contribution_extended} with input $D$, we either produce an arboreal decomposition of $D$ of width at most $3k-2$ or find a set $T \subseteq V(D)$ with $|T| = 2k-1$ such that there is no $(T, k-1)$-\balsep in $D$.
Now, by Lemma~\ref{claim:haven_instance} applied with inputs $D$, $T$, $r = k-1$, and $s=k-1$, we conclude that $D$ admits a haven of order $k$ and the result follows.
\end{proof}

Next, we show to use Theorem~\ref{theorem:first_contribution_extended} to construct a bramble in digraphs of large directed tree-width that is easier to work with than the usual construction that depends on havens (see, for instance, ~\cite[Chapter 6]{Quantitative.Graph.Theory}).


\section{Brambles and well-linked systems of paths}\label{section:cyd_grid_fpt}
Let $T$ be the set constructed by Theorem~\ref{theorem:first_contribution_extended} applied to a digraph $D$ with $n$ vertices and  $\dtw(D) \geq 3k-1$, and let $\mathcal{H}$ be the haven obtained by applying Lemma~\ref{claim:haven_instance} with input $D$ and $T$.
We remark that from $\mathcal{H}$ it is possible to construct a bramble $\mathcal{B}$ of order $\lfloor k/2 \rfloor$ and size  $|V(D)|^{\Ocal(k)}$ (see the discussion in Section~\ref{section:brambles_DGT}).
In this particular case the naive approach yields an \textsf{XP} algorithm to find a hitting set of $\mathcal{B}$ of size $k$ in \textsf{XP} time with parameter $k$, by checking all $\binom{n}{k}$ subsets $X$ of $V(D)$ with size $k$ and testing whether $X \cap V(B) \neq \emptyset$ for each $B \in \mathcal{B}$, and thus \textsf{XP} algorithms can be extracted from the constructive proofs of Propositions~\ref{lemma:directed_grid_path_hitting_set} and~\ref{lemma:directed_grid_bramble_path} assuming that these properties hold for the input brambles.
In Section~\ref{sec:brambles-large-dtw}, we show how to construct from $T$ a bramble $\mathcal{B}_T$ of order $k$ in digraphs with directed tree-width at least $3k-1$ that skips havens and is more efficient in the following two ways.

First, this construction allows us to verify whether an induced subgraph $D'$ of $D$ contains an element of $\mathcal{B}_T$ by looking only at the strong components of $D'$.
This allows us to test if a given set $X \subseteq V(D)$ is a hitting set of $\mathcal{B}_T$ in polynomial time.
Second, we show that a set $Y \subseteq V(D)$ is a minimum hitting set of $\mathcal{B}_T$ if and only if $Y$ is a solution for an appropriately defined instance of \textsc{Balanced Separator}.
Since we showed that this problem is {\sf FPT} with parameter $|T|$ (Theorem~\ref{theorem:FPT_haven}), we can compute hitting sets of $\mathcal{B}_T$ in {\sf FPT} time with parameter $\ord(\mathcal{B}_T)$.
Then, in Section~\ref{sec:finding_P_and_A} we use those results to prove stronger versions of Propositions~\ref{lemma:directed_grid_path_hitting_set} and~\ref{lemma:directed_grid_bramble_path}.



\subsection{Brambles in digraphs of large directed tree-width}
\label{sec:brambles-large-dtw}

We now define \emph{$T$-brambles} and some of its properties when $T$ is the set obtained by applying Theorem~\ref{theorem:first_contribution_extended} to a digraph $D$ with $\dtw(D) \geq 3k-1$.


\begin{definition}
Let $D$ be a digraph and $T \subseteq V(D)$ with $|T| = 2k-1$.
The \emph{$T$-bramble} $\mathcal{B}_T$ of $D$ is defined as
\[\mathcal{B}_T = \{B \subseteq D \mid B \text{ is induced, strongly connected, and } |V(B) \cap T| \geq k\}.\]
\end{definition}
\noindent Notice that $\mathcal{B}_T$ is a bramble since, as $|T|=2k-1$, any two of its element intersect.
We remark that, in general, it is possible that $\ord(\mathcal{B}_T)$ is very small: it is in fact zero if, for example, no two vertices of $T$ lay in the same strong component of $D$.
Note also that $\mathcal{B}_T$ may be empty if, for instance, any strong component of $D$ has size strictly smaller than $k$.

\begin{lemma}\label{lemma:BT_order_hitting_set}
Let $D$ be a digraph and $T$ be a $(k-1, k-1)$-linked set  of size $2k-1$ in $D$. 
Then the $T$-bramble $\mathcal{B}_T$ is a bramble of order $k$ and a set $X \subseteq V(D)$ is a hitting set of $\mathcal{B}_T$ if and only if $X$ is a $(T,k-1)$-\balsep.
\end{lemma}
\begin{proof}
Let $D$, $T$ and $\mathcal{B}_T$ be as in the statement of the lemma.
Since $|T| = 2k-1$, any set containing $k$ vertices of $T$ is a hitting set of $\mathcal{B}$.
Thus $\ord(\mathcal{B}_T) \leq k$.
Let $Z \subseteq V(D)$ with $|Z| \leq k-1$.
By definition of $(k-1, k-1)$-linked sets, $D$ does not contain any $(T,k-1)$-\balsep of size $k-1$, and hence there is a strong component $B$ of $D \setminus Z$ such that $|V(B) \cap T| \geq k$.
Since $V(B) \cap Z = \emptyset$ and $B \in \mathcal{B}_T$, we conclude that $Z$ is not a hitting set of $\mathcal{B}_T$ and therefore $\ord(\mathcal{B}_T) = k$.

For the second part of the lemma, let $X$ be a hitting set of $\mathcal{B}_T$.
Then $|V(C) \cap T| \leq k-1$ holds for every strong component $C$ of $D \setminus X$ and, by definition, $X$ is a $(T,k-1)$-\balsep.
Similarly, if $X$ is a $(T, k-1)$-\balsep then, by definition of $\mathcal{B}_T$, $X$ is a hitting set of $\mathcal{B}_T$ and the result follows.
\end{proof}

Note that we can check whether a given set $X \subseteq V(D)$ is a hitting set of $\mathcal{B}_T$ by enumerating the strong components of $D \setminus X$ and, for each such a component $C$, checking whether $|V(C) \cap T| \geq k$. This can be done in time $\Ocal(n+m)$.
For the remainder of this section, and unless stated otherwise, let $T$ be a $(k-1, k-1)$-linked set with $|T| = 2k-1$.
In what follows, we use $T$-brambles to adapt Proposition~\ref{lemma:directed_grid_bramble_path} into an \textsf{FPT} algorithm.


To prove our version of Proposition~\ref{lemma:directed_grid_bramble_path}, we start with a $T$-bramble $\mathcal{B}_T$ of order $g(k)$ (the value of $g(k)$ is specified later) in a digraph $D$ with $\dtw(D) \geq 3g(k)-1$, and then we show how to find in polynomial time a path $P(\mathcal{B}_T)$ that is a hitting set of $\mathcal{B}_T$, adapting the proof of Proposition~\ref{lemma:directed_grid_path_hitting_set} shown in~\cite[Lemma 4.3 of the full version]{Kawarabayashi:2015:DGT:2746539.2746586}.
Next, we need to show how to split $\mathcal{B}_T$ into brambles of order at least $\lceil k/2 \rceil$ whose elements are intersected by subpaths of $P(\mathcal{B}_T)$.
We do this by growing a subpath of $P'$ of $P(\mathcal{B}_T)$ iteratively while checking, on each iteration, whether the set $\mathcal{B}'_T$ of elements of $\mathcal{B}_T$ intersecting $V(P')$ is a bramble of adequate order.

We now show how our choice of $\mathcal{B}_T$ allows us to estimate the order of $\mathcal{B}'_T$ by computing the order of its ``complement bramble'' $\mathcal{B}_T \setminus \mathcal{B}'_T$, and we show how to do this procedure in {\sf FPT} time with parameter $\ord(\mathcal{B}_T)$.
These ideas are formalized by the following definitions and results.
\begin{definition}\label{def:BX}
Let $X \subseteq V(D)$ and $\mathcal{B}$ be a bramble in $D$. The \emph{restricted bramble} $\mathcal{B}(X)$ contains the elements of $\mathcal{B}$ intersecting $X$ and its \emph{complement bramble} $\overline{\mathcal{B}}(X)$ contains the elements of $\mathcal{B}$ disjoint from $X$.
Formally,
\[\mathcal{B}(X) = \{B \in \mathcal{B} \mid V(B) \cap X \neq \emptyset\},\]
\[\overline{\mathcal{B}}(X) = \{B \in \mathcal{B} \mid V(B) \cap X = \emptyset \}.\]
\end{definition}
Notice that both $\mathcal{B}(X)$ and $\overline{\mathcal{B}}(X)$ are brambles, as both are subsets of a bramble $\mathcal{B}$. Additionally, $\mathcal{B}(X)$ is disjoint from $\overline{\mathcal{B}}(X)$ and the union of a  hitting set of the former with a hitting set of the latter is a hitting set of $\mathcal{B}$.
From this remark, we have that
\begin{equation}\label{eq:bramble_order}
\ord(\mathcal{B}(X)) + \ord(\overline{\mathcal{B}}(X)) \geq \ord(\mathcal{B}),
\end{equation}
and although in general the order of $\mathcal{B}(X)$ is hard to compute, we can estimate it by knowing the order of its complement bramble $\overline{\mathcal{B}}(X)$ and $\ord(\mathcal{B})$.

Consider now the brambles $\mathcal{B}_T$, $\mathcal{B}_T(X)$, and $\overline{\mathcal{B}_T}(X)$ for some $X \subseteq V(D)$.
The following results show that hitting sets of $\overline{\mathcal{B}_T}(X)$ are exactly $(T \setminus X, k-1)$-\balsep{s} in $D \setminus X$.

\begin{lemma}\label{lemma:haven_like_bramble_x}
Let $X,Z \subseteq V(D)$ and $B$ be a strongly connected subgraph of $D$.
Then $B \in \overline{\mathcal{B}_T}(X)$ and $V(B) \cap Z = \emptyset$ if and only if $B$ is a strongly connected subgraph of $D \setminus (Z \cup X)$ with $|V(B) \cap T| \geq k$.

\end{lemma}
\begin{proof}
For the necessity, assume that $B \in \overline{\mathcal{B}_T}$ and $V(B) \cap Z = \emptyset$.
Then by the definition of $\overline{\mathcal{B}_T}(X)$, $B$ is a strongly connected subgraph of $D \setminus (Z \cup X)$ intersecting $T$ in at least $k$ vertices.

For the sufficiency, assume that $B$ is a strongly connected subgraph of $D \setminus (Z \cup X)$ containing at least $k$ vertices of $T$.
Then $B \in \overline{\mathcal{B}_T}(X)$ by the definition of $\overline{\mathcal{B}_T}(X)$ and the lemma follows since it is disjoint from $Z \cup X$.
\end{proof}

\noindent  The contrapositive of Lemma~\ref{lemma:haven_like_bramble_x} characterizes hitting sets of $\overline{\mathcal{B}_T}(X)$.
\begin{corollary}\label{corollary:haven_like_bramble_x_contrapositive}
Let $X,Z \subseteq V(D)$. $Z$ is a hitting set of $\overline{\mathcal{B}_T}(X)$ if and only if $Z$ is a $(T \setminus X, k-1)$-\balsep in $D \setminus X$.
\end{corollary}
Therefore, we can decide whether $\ord(\overline{\mathcal{B}_T}(X)) \leq s$ by testing whether $D$ admits a $(T\setminus X,k-1)$-\balsep of size $s$.
The following result is a direct consequence of Theorem~\ref{theorem:FPT_haven} and Corollary~\ref{corollary:haven_like_bramble_x_contrapositive}.
\begin{corollary}\label{corollary:hitting_set_FPT}
For any $X \subseteq V(D)$, there is an algorithm running in time $\RunTimeAlg$ that decides whether $\ord(\overline{\mathcal{B}_T}(X)) \leq s$.
\end{corollary}

Next, we show how to find such a path $P(\mathcal{B}_T)$ as described above and a well-linked set $A$ of size roughly $\sqrt{2k}$ that is contained in $V(P(\mathcal{B}_T))$.



\subsection{Finding \texorpdfstring{$P(\mathcal{B}_T)$}{P} and \texorpdfstring{$A$}{A}}\label{sec:finding_P_and_A}

The proof of the next lemma is an adaptation of the proof of~\cite[Lemma 4.3 of the full version]{Kawarabayashi:2015:DGT:2746539.2746586} to our scenario.
We exploit the fact that we can check whether a given set of vertices is a hitting set of $\mathcal{B}_T$ in polynomial time: by Lemma~\ref{lemma:BT_order_hitting_set}, a set $X \subseteq V(D)$ is a hitting set of $\mathcal{B}_T$ if and only if $X$ is a $(T, k-1)$-\balsep, and we can check if a given set $X$ is a $(T, k-1)$-\balsep by enumerating the strong components of the input digraph.
\begin{lemma}\label{lemma:bramble_path_polytime}
Let $D$ be a digraph, let $T$ be a $(k-1, k-1)$-linked set of size $2k-1$, and consider the $T$-bramble $\mathcal{B}_T$.
There is an algorithm running in time $\Ocal(n ( n + m))$ that produces a path $P$ that is a hitting set of $\mathcal{B}_T$.
\end{lemma}
\begin{proof}


If $\ord(\mathcal{B}_T) \geq 1$, then there is an element $B \in \mathcal{B}_T$ and a strong component $C$ of $D$ such that $V(B) \subseteq V(C)$ and, by the definition of $\mathcal{B}_T$, we know that $D[V(C)] \in \mathcal{B}_T$.
Define $B_1 = D[V(C)]$, let $v_1$ be any vertex of $B_1$, and define $P_1$ as the path containing only the vertex $v_1$ and $V(P_0) = \emptyset$.
We proceed to grow a path by iterating from $P_1$ to $P_{k'}$ where they all start from $v_1$, each $P_i$ with $i \geq 2$ contains $P_{i-1}$, and $P_{k'}$ is a hitting set of $\mathcal{B}_T$.
Throughout our process, we maintain a collection of elements $B_i \in \mathcal{B}_T$ such that $V(P_i)$ intersects $V(B_i)$ only in the last vertex $v_i$ of $P_i$.
Since $|V(P_1)| = 1$ and $v_1 \in T \subseteq V(B_1)$, this condition trivially holds for $P_1$.
Assume now that $i$ paths have been chosen this way, with $i \geq 1$.

Consider the last vertex $v_i$ of the path $P_i$ and the element $B_i$ of $\mathcal{B}_T$ with $V(P_i) \cap V(B_i) = \{v_i\}$.
By Lemma~\ref{lemma:BT_order_hitting_set}, $V(P_i)$ is a hitting set of $\mathcal{B}_T$ if and only if $V(P_i)$ is a $(T,k-1)$-\balsep, and this can be tested in time $\Ocal(n+m)$ by enumerating all strong components of $D \setminus V(P_i)$.
If $V(P_i)$ is a hitting set of $\mathcal{B}_T$, then we terminate the algorithm returning $P_i$. 
Otherwise, $V(P_i)$ is not a $(T, k-1)$-\balsep and thus there is a strong component $F$ of $D \setminus V(P_i)$ with $|V(F) \cap T| \geq k$.
Therefore, $D[V(F)]$ is an element of $\mathcal{B}_T$ whose vertices are disjoint from $V(P_i)$ and we choose $B_{i+1} = D[V(F)]$.

Since $\mathcal{B}_T$ is a bramble, we can find a path $P'$ from $v_i \in V(P_i) \cap V(B_i)$ to a vertex $v_{i+1} \in B_{i+1}$ in $D[V(B_i) \cup V(B_{i+1})]$ such that $V(P') \cap V(B_{i+1}) = \{v_{i+1}\}$.
Moreover, $v_i$ is the only vertex of $P_i$ in $B_i$ and thus the path $P'$ does not contain any vertex in $V(P_i) \setminus \{v_i\}$.
Now, let $P_{i+1}$ be the path obtained from $P_i$ by appending $P'$.
By our choice of $P'$, we know that only the last vertex $v_{i+1}$ of $P_{i+1}$ is in $V(B_{i+1})$, as desired, and $V(P_{i+1})$ hits strictly more elements of $\mathcal{B}_T$ than $V(P_i)$.
We repeat the aforementioned procedure now considering the vertex $v_{i+1}$, the path $P_{i+1}$, and the element $B_{i+1}$ of $\mathcal{B}_T$.

Since we can enumerate the strong components of a subgraph of $D$ in time $\Ocal(n+m)$ (see, for instance, \cite[Chapter 6]{Graph.Theory}), at the $i$-th iteration we can find $B_{i+1}$, the path $P_{i+1}$, and the vertex $v_{i+1}$ in time $\Ocal(n+m)$.
Finally, the procedure eventually terminates as $|V(P)| \leq n$ and thus the bound on the running time follows.
\end{proof}


For the remainder of this section, we assume that $g(k) = (k+1)(\lfloor k/2 \rfloor +1) - 1$, that $D$ is a digraph containing a $(g(k)-1, g(k)-1)$-linked set $T$ of size $2g(k)-1$, consider the $T$-bramble $\mathcal{B}_T$, and fix $P$ to be the path received by applying Lemma~\ref{lemma:bramble_path_polytime} with inputs $D$, $T$, and $\mathcal{B}_T$.
To prove Theorem~\ref{theorem:path_sec_contribution}, we use the following definition.
\begin{definition}[$(i)$-split]\label{definition:i-splits}
An \emph{$(i)$-split} $\mathcal{S}$ of $P$ is a collection formed by a set $\{Q_j \mid j \in [i]\}$ of subpaths of $P$, a subpath $P_i$ of $P$, a set of brambles $\{\mathcal{B}_j \mid j \in [i]\}$, a set of vertices $\{a_j \mid j \in [i]\}$, and a set of vertices $X_i$ such that
\begin{enumerate}
\item for $j \in [i]$, vertex $a_j$ is the successor in $P$ of the last vertex of $Q_j$, and, if $j \leq i-1$, the first vertex of $Q_{j+1}$ is the successor in $P$ of vertex $a_j$,
\item for $j \in [i]$, $\ord(\mathcal{B}_j) \geq \lfloor k/2 \rfloor$,
\item for $j \in [i]$, $\mathcal{B}_j \subseteq \mathcal{B}_T$ and $V(Q_j)$ is a hitting set of $\mathcal{B}_j$,
\item $P_i$ is the subpath of $P$ from the successor in $P$ of the last vertex of $Q_i$ to the last vertex of $P$, and
\item $X_i = \bigcup_{j \in [i]}(V(P_j) \cup \{a_j\})$, and
\item the order of $\overline{\mathcal{B}_T}(X_i)$ satisfies
$$\ord(\overline{\mathcal{B}_T}(X_i)) \geq g(k) - i \left(\left\lfloor\frac{k}{2}\right\rfloor +1 \right).$$
\end{enumerate}
\end{definition}
See Figure~\ref{fig:split_bramble_set_X} for an example of a $(2)$-split.
We remark that a $(0)$-split for $P$ consists only of the path $P_0$ with $P_0 = P$ and the empty set $X_0$.
\begin{figure}[!ht]
\centering
\scalebox{.45}{
\begin{tikzpicture}
\node(P-X) at (0,0) {};
\node[scale=.8, blackvertex, label=90:{\huge $a_{2}$}] (P-a) at ($(P-X) + (4,0)$) {};

\node[scale=.8, blackvertex] (P-p2out) at ($(P-a) + (-2,0)$) {};
\node[scale=.8, blackvertex] (P-p2in) at ($(P-p2out) + (-3,0)$) {};
\node[scale=.8, blackvertex, label=90:{\huge $a_1$}] (P-a1) at ($(P-p2in) + (-2,0)$) {};
\node[scale=.8, blackvertex] (P-p1out) at ($(P-a1) + (-2,0)$) {};
\node[scale=.8, blackvertex] (P-p1in) at ($(P-p1out) + (-3,0)$) {};

\draw[-{Latex[length=3mm, width=2mm]}, shorten >= .1cm] (P-p1in) -- (P-p1out) node[midway,above,yshift=.5cm] {\huge $Q_1$};
\foreach \i in {1,...,4}{
  \node[blackvertex, scale=.5] (P-\i) at ($(P-p1in) + (\i *0.5, 0)$) {};
}

\draw[-{Latex[length=3mm, width=2mm]}, shorten >= .1cm] (P-p2in) -- (P-p2out)  node[midway,above,yshift=.5cm ] {\huge $Q_2$};
\foreach \i in {1,...,4}{
  \node[blackvertex, scale=.5] (P-\i) at ($(P-p2in) + (\i *0.5, 0)$) {};
}

\draw[-{Latex[length=3mm, width=2mm]}, shorten >= .1cm] (P-p1out) -- (P-a1);
\draw[-{Latex[length=3mm, width=2mm]}, shorten >= .1cm] (P-a1) -- (P-p2in);
\draw[-{Latex[length=3mm, width=2mm]}, shorten >= .1cm] (P-p2out) -- (P-a);

\draw[shorten >= .2cm, shorten <= .2cm] ($(P-p1in)+ (0,-1)$) -- ($(P-p1in)+ (0,-2)$);
\draw[shorten >= .2cm, shorten <= .2cm] ($(P-a)+ (0,-1)$) -- ($(P-a)+ (0,-2)$);
\node (P-n) at ($(P-p1in) + (6,-1.5)$) {\huge $X_2$};
\draw[dashed, shorten >= .2cm, shorten <= .2cm] ($(P-p1in) + (0,-1.5)$) -- (P-n);
\draw[dashed, shorten >= .2cm, shorten <= .2cm] (P-n) -- ($(P-a) + (0,-1.5)$);

\node[scale=.8, blackvertex] (P-b) at ($(P-a) + (2,0)$) {};
\node[scale=.8, blackvertex] (P-c) at ($(P-b) + (10,0)$) {};
\draw[-{Latex[length=3mm, width=2mm]}, shorten >= .1cm] (P-a) -- (P-b);
\draw[-{Latex[length=3mm, width=2mm]}, shorten >= .1cm] (P-b) -- (P-c);

\draw[shorten >= .2cm, shorten <= .2cm] ($(P-b)+ (0,-1)$) -- ($(P-b)+ (0,-2)$);
\draw[shorten >= .2cm, shorten <= .2cm] ($(P-c)+ (0,-1)$) -- ($(P-c)+ (0,-2)$);
\node (P-n) at ($(P-b) + (5,-1.5)$) {\huge $V(P_2) = V(P) \setminus X_2$};
\draw[dashed, shorten >= .2cm, shorten <= .2cm] ($(P-b) + (0,-1.5)$) -- (P-n);
\draw[dashed, shorten >= .2cm, shorten <= .2cm] (P-n) -- ($(P-c) + (0,-1.5)$);

\foreach \i in {1,...,5}{
  \node[draw,circle,align=center,scale=4] (P-\i) at  ($(P-b) + (1.6*\i,0)$) {};
  }

\draw [decorate,decoration={brace,amplitude=10pt},xshift=0pt,yshift=0]
($(P-b) + (1,0.8)$) -- ($(P-c) + (-1.5,0.8)$) node [black,midway,yshift=.85cm] {\huge $\overline{\mathcal{B}_T}(X_2)$};

\end{tikzpicture}
}%
\caption{Illustration of a $(2)$-split of $P$. A circle represents an element of the bramble $\overline{\mathcal{B}}(X_2)$.}
\label{fig:split_bramble_set_X}
\end{figure}
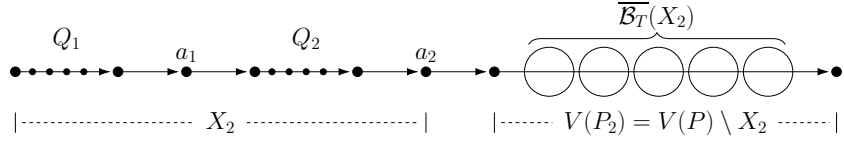

Now, the proof of Theorem~\ref{theorem:path_sec_contribution} follows three steps.
First, Lemma~\ref{lemma:split_implies_well_linked} states that the set of vertices $\{a_1, \ldots, a_i\}$ of an $(i)$-split of $P$ is well-linked when  $i \leq k$.
Thus our goal is to construct a $(k)$-split of $P$.
Then, Lemma~\ref{lemma:split_iteration} states that, for $i \geq 0$, we can construct an $(i+1)$-split of $P$ from an $(i)$-split of $P$ in \textsf{FPT} time if $\ord(\overline{\mathcal{B}}(X_i))$ is large enough.
Finally, the proof of Theorem~\ref{theorem:path_sec_contribution} starts from a $(0)$-split of $P$ and iterates Lemma~\ref{lemma:split_iteration} until a $(k)$-split is constructed.

\begin{lemma}\label{lemma:split_implies_well_linked}
Let $\mathcal{S}_i$ be an $(i)$-split of $P$ with $i \in [k]$.
Then the set $A$ with $A = \{a_1, \ldots, a_i\}$ is well-linked in $D$.
\end{lemma}
\begin{proof}
Let $X$ and $Y$ be disjoint subsets of $A$ such that $|X| = |Y| = r$ for some $r \in [i]$.
Suppose, by contradiction, that there is no set of $r$ pairwise internally disjoint paths from $X$ to $Y$ in $D$.
Then, by Menger's Theorem, there is an $(X,Y)$-separator $S \subseteq V(D)$ such that $|S| \leq r-1$.

Let $Q_{i+1} = P_i$ and $\mathcal{B}_{i+1} = \overline{\mathcal{B}}(X_i)$.
By the definition of $(i)$-splits and our choice of $Q_{i+1}$, we know that  for every $a_j \in A$ with $j \in [i]$, $Q_{j}$ is a path ending on the vertex occurring in $P$ before $a_j$, and $Q_{j+1}$ is a path starting on the first vertex occurring in $P$ after $a_j$ (see Figure~\ref{fig:split_bramble_set_X} for an example when $i=2$).
Moreover, we have
$$\ord(\mathcal{B}_{i+1}) \geq g(k) - i\left(\left\lfloor \frac{k}{2} \right\rfloor + 1\right)$$
which implies that $\ord(\mathcal{B}_{i+1}) \geq \lfloor k/2 \rfloor$ since $i \leq k$.

As $|S| \leq r-1 \leq \lfloor k/2 \rfloor  - 1$ there is a $j \in [i-1]$ such that $a_j \in X \setminus S$ and $S \cap V(Q_{j+1}) = \emptyset$.
Furthermore, since $S$ is not large enough to be a hitting set of $\mathcal{B}_{j+1}$, there must be $B \in \mathcal{B}_{j+1}$ such that $S \cap V(B) = \emptyset$.
Similarly, there are $a_\ell \in Y \setminus S$ and $B' \in \mathcal{B}_\ell$ such that $S \cap V(Q_\ell) = \emptyset$ and $S \cap V(B') = \emptyset$.

By choice,  clearly $V(Q_{j+1})$ is a hitting set of $\mathcal{B}_{j+1}$.
Since $S$ is disjoint from $V(Q_{j+1}) \cup V(B)$ and $V(B)$ induces a strongly connected subgraph of $D$, we conclude that there is in $D \setminus S$ a path from $a_{j}$ to any vertex in $V(B)$.
Similarly, there is a path from any vertex in $V(B')$ to $a_{\ell}$ in $D \setminus S$.
Finally, since every pair of elements in $\mathcal{B}_T$ intersect, we conclude that there is a path in $D \setminus S$ from $a_{j}$ to $a_{\ell}$ using the path $Q_{j+1}$, the vertices in $V(B) \cup V(B)'$, and the path $Q_{\ell}$.
This contradicts our choice of $S$, and thus we conclude that every $(X,Y)$-separator in $D$ must have size at least $r$, and the result follows by Menger's Theorem.
\end{proof}

\begin{lemma}\label{lemma:split_iteration}
Let $\mathcal{S}_i$ be an $(i)$-split of $P$ with $i \leq k-1$.
Then in time $2^{\Ocal(k^2 \log k)}\cdot n^{\Ocal(1)}$ we can construct an $(i+1)$-split of $P$.
\end{lemma}
\begin{proof}
For a digraph $F$, for the sake of notational simplicity, we abbreviate --recall Definition~\ref{def:BX}-- $\mathcal{B}(V(F))$ and $\overline{\mathcal{B}}(V(F))$ as $\mathcal{B}(F)$ and $\overline{\mathcal{B}}(F)$, respectively, and write $\mathcal{B}(v)$ and $\overline{\mathcal{B}}(v)$ (omitting the braces) for $v \in V(F)$.
Let $\mathcal{B'} = \overline{\mathcal{B}}(X_i)$.

The goal is to construct a subpath $Q_{i+1}$ of $P$ starting on the first vertex of $P$ appearing after the vertex $a_i$ (or simply the first vertex of $P$ if $i=0$) such that
$$\ord(\mathcal{B'}(Q_{i+1})) \geq \left\lfloor\frac{k}{2}\right\rfloor.$$
That is, the order of the bramble containing the elements of $\mathcal{B}$ which are disjoint from $X_i$ while intersecting $V(Q_{i+1})$ is at least $\lfloor k/2 \rfloor$.
We start with $V(Q_{i+1}) = \emptyset$.
By Inequality~\ref{eq:bramble_order}, we have that
$$\ord(\mathcal{B}'(Q_{i+1})) \geq \ord(\mathcal{B}') - \ord(\overline{\mathcal{B}'}(Q_{i+1}))$$
at any point of the procedure.
Now, we iteratively grow $Q_{i+1}$, adding one vertex at a time while testing, at each newly added vertex, whether
$$\ord(\overline{\mathcal{B}'}(Q_{i+1})) \leq g(k) - i\left(\left\lfloor\frac{k}{2}\right\rfloor + 1\right) - 1 - \left\lfloor \frac{k}{2} \right\rfloor.$$
Observe that, when $V(Q_{i+1}) = \emptyset$, we have $\overline{\mathcal{B}'}(Q_{i+1}) = \mathcal{B}'$ and thus
$$\ord(\overline{\mathcal{B}'}(Q_{i+1})) \geq g(k) - i\left(\left\lfloor \frac{k}{2}\right\rfloor +1\right) >  g(k) - i\left(\left\lfloor \frac{k}{2}\right\rfloor +1\right) - \left\lfloor\frac{k}{2}\right\rfloor.$$

As $\mathcal{B'} = \overline{\mathcal{B}}(X_i)$, we have $\overline{\mathcal{B}'}(Q_{i+1}) = \overline{\mathcal{B}}(X_i \cup V(Q_{i+1}))$ and thus, by Corollary~\ref{corollary:hitting_set_FPT}, we can test whether $\ord(\overline{\mathcal{B}'}(Q_{i+1})) \leq s$ in time $2^{\Ocal(k^2 \log k)}\cdot n^{\Ocal(1)}$ for any $s \in [g(k)]$ since $g(k) = \Ocal(k^2)$.

On a negative answer, we add to $Q_{i+1}$ the first vertex of $P$ not contained in $V(Q_{i+1}) \cup X_i$ and repeat the test.
On the first time we obtain a positive answer to this test, we set $\mathcal{B}_{i+1} = \mathcal{B}'(Q_{i+1})$, define $a_{i+1}$ to be the first vertex appearing in $P$ after the last vertex of $Q_{i+1}$, and stop the procedure.
In this case, we have that $\ord(\mathcal{B}_{i+1}) \geq \lfloor k/2 \rfloor$ and since $\ord\left(\overline{\mathcal{B}'}(Q_{i+1})\right)$ can decrease by at most one each time we increase by one the size of $V(Q_{i+1})$, this procedure actually ends with $\ord(\overline{\mathcal{B}'}(Q_{i+1})) = g(k) - i(\lfloor k/2 \rfloor + 1) - \lfloor k/2 \rfloor$.

Now, we define $X_{i+1} = X_i \cup V(Q_{i+1}) \cup \{a_{i+1}\}$ and $P_{i+1}$ to be the subpath of $P$ with $V(P_{i+1}) = V(P) \setminus X_{i+1}$.
Finally, let $\mathcal{B}^* = \overline{\mathcal{B}'}(Q_{i+1})$.
Then by Inequality~\ref{eq:bramble_order}, 
$$\ord(\mathcal{B}^*(P_{i+1})) \geq \ord(\mathcal{B}^*) - \ord(\overline{\mathcal{B}^*}(P_{i+1}))$$
and observing that $\mathcal{B}^*(P_{i+1}) = \overline{\mathcal{B}}(X_{i+1})$, we conclude that
$$\ord(\overline{\mathcal{B}}(X_{i+1})) \geq g(k) - i\left(\left\lfloor \frac{k}{2}\right\rfloor +1\right) - \left\lfloor\frac{k}{2}\right\rfloor - 1 = g(k) - (i+1)\left(\left\lfloor \frac{k}{2}\right\rfloor +1\right),$$
as required, since $\overline{\mathcal{B}^*}(P_{i+1}) = \overline{\mathcal{B}^*}(a_{i+1})$ and thus $\ord(\overline{\mathcal{B}^*}(P_{i+1})) \leq 1$.
Then, we output the $(i+1)$-split $\mathcal{S}_{i+1}$ of $P_i$ formed by the sequence of paths $Q_1, \ldots, Q_{i+1}$, the path $P_{i+1}$, the sequence of brambles $\mathcal{B}_1, \ldots, \mathcal{B}_{i+1}$, the set of vertices $\{a_1, \ldots, a_{i+1}\}$, and the set of vertices $X_{i+1}$.
\end{proof}

We remark that the bramble $\overline{\mathcal{B}'}(Q_{i+1})$ is used only in the proof of Lemma~\ref{lemma:split_implies_well_linked} and thus we do not need to maintain it during the algorithm.
However, if we want to store this information, it suffices to maintain the set $T$, the set $X_i$, and the path $Q_{i+1}$ since the bramble $\overline{\mathcal{B}'}(Q_{i+1})$ is equal to the bramble $\mathcal{B}(Q_{i+1})$ in the digraph $D \setminus X_i$.
We are now ready to prove Theorem~\ref{theorem:path_sec_contribution}.
\secondmain*
\begin{proof}
By Lemma~\ref{lemma:BT_order_hitting_set}, the $T$-bramble $\mathcal{B}_T$ has order $g(k)$ and, by Lemma~\ref{lemma:bramble_path_polytime}, we can find a path $P$ that is a hitting set of $\mathcal{B}_T$ in polynomial time.
We start with a trivial $(0)$-split $\mathcal{S}_0$ of $P$ where $P_0 = P$ and $X_0 = \emptyset$.

For $i \in \{0, \ldots, k-1\}$, we apply Lemma~\ref{lemma:split_iteration} with input $\mathcal{S}_i$ to obtain an $(i+1)$-split $\mathcal{S}_{i+1}$ of $P$ in time $2^{\Ocal(k^2 \log k)}\cdot n^{\Ocal(1)}$.
After the last iteration, we obtained a $(k)$-split $\mathcal{S}_k$ of $P$ and, by Lemma~\ref{lemma:split_implies_well_linked}, the set of vertices $\{a_1, \ldots, a_k\}$ of $\mathcal{S}_k$ is well-linked in $D$ and all such vertices are in $V(P)$, as desired.
\end{proof}

By following the remainder of the proof of the Directed Grid Theorem~\cite{Kawarabayashi:2015:DGT:2746539.2746586}, which yields {\sf FPT} algorithms for all the remaining steps (see Section~\ref{section:finding_a_cylindrical_grid}), we can validate Corollary~\ref{theorem:second_contribution}.


\section{Concluding remarks}\label{section:concluding_remarks}
The main consequence of our results is an {\sf FPT} algorithm with parameter $k$ that either produces an arboreal decomposition of width at most $f(k)$ for a digraph $D$ or constructs a cylindrical grid of order $k$ as a butterfly minor of $D$, for some computable function $f(k)$.
This is achieved by adapting some of the steps used in the proof of the Directed Grid Theorem from Kawarabayashi and Kreutzer~\cite{Kawarabayashi:2015:DGT:2746539.2746586}.

For the first possible output of this algorithm, we improve on a result from~\cite{Johnson.Robertson.Seymour.Thomas.01} by providing an {\sf FPT} algorithm with parameter $k$ that either produces an arboreal decomposition of a digraph $D$ with width at most $3k-2$, or concludes that $D$ has a haven of order $k$.
As a tool to prove this result, we consider a generalization of the problem of finding balanced separators in digraphs (we remind the reader that our definition of balanced separators extends the classical definition that can be found, for example, in~\cite{Quantitative.Graph.Theory}) and show how to solve it in {\sf FPT} time with parameter $|T|$.
Since in the undirected case balanced separators are strongly related to the tree-width of undirected graphs, and the only result for balanced separators in the directed case considered only a relaxed version of the problem (see~\cite[Chapter 6]{Quantitative.Graph.Theory}), we consider this result to be of its own interest.

Although it is possible to construct a bramble $\mathcal{B}$ of order $\lfloor k/2 \rfloor$ from a haven of order $k$, this construction is not {\sl efficient} in general, in the sense that we must go through all elements of $\mathcal{B}$ to verify whether a given set $X$ is a hitting set of $\mathcal{B}$.
Motivated by this, we consider a definition of brambles, which we call $T$-brambles, which naturally occur in digraphs of large directed tree-width that are better to work with in a number of ways.
For instance, by reducing to the problem of computing $(T, r)$-balanced separators for $T$, we show how to compute hitting sets of $T$-brambles in {\sf FPT} time when parameterized by $|T|$.

We use our results for $T$-brambles in digraphs of large tree-width to show how to find, in ${\sf FPT}$ time with parameter $k$, a path that is a hitting set of a $T$-bramble $\mathcal{B}_T$ of order $(k+1)(\lfloor k/2 \rfloor +1)$ and a well-linked set of size $k$ that is contained in this path.
This is the second step that we change in the proof of the Directed Grid Theorem~\cite{Kawarabayashi:2015:DGT:2746539.2746586}.
From this point forward, the remaining steps in the proof yield {\sf FPT} algorithms.

Kreutzer and Ordyniak~\cite{KREUTZER20114688} and Ganian et al.~\cite{GANIAN201488} showed that many important problems in digraphs remain hard when restricted to digraphs of bounded directed tree-width.
In particular, Kreutzer and Ordyniak~\cite{KREUTZER20114688} showed that the \textsc{Directed Feedback Vertex Set} (\textsc{DFVS}) problem is \textsf{NP}-complete even when restricted to digraphs of directed tree-width at most five.
However, some open problems in digraphs may benefit from an approach resembling Bidimensionality using our \textsf{FPT} algorithm for the Directed Grid Theorem.
For example, Bezáková et al.~\cite{Bezakova.Curticapean.Dell.Fomin.2016} asked whether the \textsc{Longest Detour} problem in digraphs could be solved by using the Directed Grid Theorem.
To provide more potential applicability of our results, we briefly discuss the parameterized tractability of the \textsc{DFVS} problem.

Chen et al.~\cite{Chen.Liu.Lu.Sullivan.Razgon.2008} provided an algorithm running in time $2^{\Ocal(k \log k)} \cdot n^{\Ocal(1)}$ for the \textsc{DFVS} problem, where $k$ is the size of the solution.
Bonamy et al.~\cite{10.1007/978-3-030-00256-5_6} showed that, when parameterized by the tree-width $t$ of the underlying graph, \textsc{DFVS} is solvable in time $2^{\Ocal(t \log t)} \cdot n^{\Ocal(1)}$ in general digraphs and the dependency on the parameter is improved to $2^{\Ocal(t)}$ when restricted to planar digraphs.
When parameterized by the feedback vertex set number of the underlying graph, Bergougnoux et al.~\cite{benjamin_et_al:LIPIcs:2017:8112} showed that \textsc{DFVS} admits a polynomial kernel in general digraphs, and a linear kernel in digraphs that are embeddable on surfaces of bounded genus.

On the one hand, \textsc{DFVS} remains hard even when restricted to digraphs of directed tree-width at most five~\cite{KREUTZER20114688}, but on the other hand both of the aforementioned parameters related to the underlying graph are individually stronger than the directed tree-width of the input digraph and, by the Directed Grid Theorem~\cite{Kawarabayashi:2015:DGT:2746539.2746586}, every positive instance of \textsc{DFVS} parameterized by the size $k$ of the solution occurs in a digraph of bounded directed tree-width: since a cylindrical grid of order $r$ contains a set of $r$ vertex-disjoint  cycles and butterfly contractions do not generate new paths, the minimum size of a feedback vertex set of a digraph $D$ is at least the order of the largest cylindrical grid that is as a butterfly minor of $D$.
Now, by Corollary~\ref{theorem:second_contribution}, in \textsf{FPT} time with parameter $k$ we can either find a certificate that the considered instance of \textsc{DFVS} is negative (a cylindrical grid of order $k+1$ that is a butterfly minor of the input digraph), or produce an arboreal decomposition of the input digraph of width at most $f(k)$, for some computable function $f: \mathbb{N} \to \mathbb{N}$.

Thus, it is sensible to ask whether similar or improved results for \textsc{DFVS} (when parameterized by the tree-width or the feedback vertex set number of the underlying graph, as previously mentioned) can be proved if we consider that the input digraph has bounded directed tree-width, since by the above discussion we can restrict instances of \textsc{DFVS} to this class of digraphs.

One could also consider the tractability of hard problems in digraphs of bounded directed tree-width under stronger parameterizations.
For example, Lopes and Sau\xspace\cite{LoSa20} recently showed that a relaxation for the \textsc{Directed Disjoint Paths} problem, a notoriously hard problem in digraphs, admits a kernelization algorithm for some choices of parameters.
In this spirit, it seems plausible that combining directed tree-width with other parameters may lead to \textsf{FPT} algorithms for hard problems, and in this context the \textsf{FPT} algorithm presented in this paper may become handy.

It is worth mentioning that Giannopoulou et al.~\cite{GiannopoulouKKK20} recently provided an analogous version of the Flat Wall Theorem~\cite{ROBERTSON199565} for directed graphs, which may have interesting algorithmic applications when combined with our results.

Finally, the attempts to obtain a Bidimensionality theory for directed graphs, such as the one presented by Dorn et al.~\cite{DornFLRS13}, are so far less satisfying that the undirected version, from the point of view of generality and efficiency of the obtained algorithms. We hope that our \textsf{FPT} version of the Directed Grid Theorem will have a relevant role in an eventual Bidimensionality theory for directed graphs.



\bibliographystyle{abbrv-initials}
\bibliography{references}
\end{document}